\documentclass[journal]{IEEEtran}
\IEEEoverridecommandlockouts

\usepackage{lineno}

\usepackage{amsthm}
\newtheorem{remark}{Remark}
\newtheorem{lemma}{Lemma}
\newtheorem{example}{Example}

\usepackage{tabularx}
\usepackage{array}  % in the preamble
\newcolumntype{C}{>{\centering\arraybackslash}X}

\usepackage{cite}

\usepackage[colorlinks=true, linkcolor=black, citecolor=black, urlcolor=black]{hyperref}

\usepackage{bm}
\usepackage{multicol}
\usepackage{scalerel}
\usepackage{tikz, pgfplots}
\usetikzlibrary{decorations.pathmorphing}
\pgfplotsset{compat=1.18}
\usetikzlibrary{positioning,calc,shapes.geometric, patterns, decorations.markings}
\usetikzlibrary{through,backgrounds}
\usepackage{rotating}

\usepackage{microtype} % \textls

\usepackage{subfig}
\usepackage{mathrsfs}
\usepackage{makecell}
\usepackage{booktabs}
\usepackage{ragged2e}
\usepackage{amsfonts}
\usepackage{mathtools}

\usepackage{multirow}
\usepackage{algorithmic,xpatch}
\usepackage[linesnumbered,ruled,vlined,lined,boxed,commentsnumbered,longend,algoruled]{algorithm2e}
\usepackage{flushend}
\newcommand{\ie}{i.e.}
\newcommand{\emphsec}[1]{\emph{#1}:}

\newcommand{\ra}{\rightarrow}
\DeclareMathOperator{\prob}{Pr}
\DeclareMathOperator{\erfc}{erfc}

\usepackage{etoolbox}

\newcommand{\vc}[1]{\mathbf{#1}}

\newcommand{\norm}[1]{\lVert #1 \rVert}
\newcommand{\iid}{i.i.d.}
\newcommand{\eqdef}{\stackrel{\Delta}{=}}
\newcommand{\E}{\mathbb{E}}
\newcommand{\eg}{e.g.,}
\newcommand{\frommton}[2]{[#1: #2]}
\newcommand{\fromoneton}[1]{[#1]}
\DeclareMathOperator{\argmin}{argmin}

\newcommand{\matdet}[1]{\mathrm{#1}}

\newcommand{\FF}{\mathbb{F}_2}

\newcommand{\QF}{\text{QF}}
\newcommand{\ber}{\text{BER}}

\DeclareUnicodeCharacter{2212}{-}

\usepackage{hhline}
\usepackage{longtable}

\usepackage{flushend}
\usepackage{comment}
\usepackage{accents}
\newcommand{\ubar}[1]{\underaccent{\bar}{#1}}

\title{Adaptive Learned Belief Propagation for Decoding Error-Correcting Codes}

\author{
Alireza Tasdighi and Mansoor Yousefi
\thanks{Alireza Tasdighi and Mansoor Yousefi are with the Communications and Electronics Department of Télécom Paris (Institut Polytechnique de Paris), 
Paris, France. E-mails:\{alireza.tasdighi,yousefi\}@telecom-paris.fr.
\\
This work has received funding from the European Research Council under the European Union's 
Horizon 2020 research and innovation programme, Grant Agreement No. 805195.
}
}

\begin{document}
\maketitle

%%%%%%%%%%%%%%%%%%%%%%%%%%%%%%%%%%
%%%% ABSTRACT
%%%%%%%%%%%%%%%%%%%%%%%%%%%%%%%%%%

\begin{abstract}  
Weighted belief propagation (WBP) for the decoding of   linear block codes is considered.
In WBP, the Tanner graph of the code is unrolled with respect to the iterations of the belief propagation decoder.
Then, weights are assigned to the edges of the resulting recurrent network and optimized offline using a training dataset.
The main contribution of this paper is an adaptive WBP where the weights of the decoder
are determined for each received word. Two variants of this decoder are investigated.
In the parallel WBP decoders, the weights take values in a discrete set. A number of WBP decoders are run in
parallel to search for the best sequence- of weights in real time.
In the two-stage decoder,
a small neural network is used to dynamically determine the weights of the WBP decoder for each received word.
The proposed adaptive decoders demonstrate significant improvements over the static counterparts in two applications.
In the first application, Bose--Chaudhuri--Hocquenghem, polar and quasi-cyclic low-density parity-check (QC-LDPC) codes are used over an additive white Gaussian
noise channel.
The results indicate that the adaptive WBP achieves bit error rates (BERs) up to an order of magnitude less than
the BERs of the static WBP at about the same decoding complexity, depending on the code, its rate, and the signal-to-noise ratio.
The second application is a concatenated code designed for a long-haul nonlinear optical fiber channel 
where the inner code is a QC-LDPC code and the outer code is a spatially coupled LDPC code.
In this case, the inner code is decoded using an adaptive WBP, while the outer code is decoded using the sliding window decoder and
static belief propagation.
The results show that the adaptive WBP provides
a coding gain of 0.8 dB compared to the neural normalized min-sum decoder, with about the same computational
complexity and decoding latency.

\end{abstract}


%%%%%%%%%%%%%%%%%%%%%%%%%%%%%%%%%%
%%%% SECTION I: Introduction
%%%%%%%%%%%%%%%%%%%%%%%%%%%%%%%%%%

\section{Introduction}
\label{Intro}
 Neural networks (NNs) have been widely studied to improve communication systems.
The ability of NNs to learn from data and model complex relationships makes them indispensable tools for tasks such as equalization,
monitoring, modulation classification, and beamforming \cite{Survey2021AI}.
While NNs have also been considered for decoding   error-correcting codes for quite some time
\cite{J.Bruck89,G.Zeng89,W.R.Caid90,tseng1993,G.Marcone95,X.A.Wang96,G.Tallini95,M.Ibnkahla2000,haroon2013},
interest in this area has surged
significantly in recent years due to advances in NNs and their widespread commercialization
~\cite{E.Nachmani2016,T.Gruber2017, H.Kim2018,vasic2018learning,A.Bennatan2018, E.Nachmani18, Lugosch18,M.Lian2019,
jiang2019deepturbo,carpi2019RL,Q.Wang2020,  huang2019ai,Be2020,xu_deep_2020,A.Buchberger2020,J.Dai2021,
T.Tonnellier2021,L.Wang2021,habib2021belief,nachmani2021autoregressive,E.Nachmani2022,cammerer2020trainable,
cammerer2022graph,Y.Choukroun2022,jamali2022productae,dorner2022learning, li2023bottom,Q.Wang2023,wang2023ldpc,
clausius2023component,Y.Choukroun2024a,Y.Choukroun2024b,clausius2024graph,S.Adiga2024,kim2024neural,ninkovic2024decoding}.

Two categories of neural decoders may be considered.
In model-agnostic decoders, the NN has a general architecture  independent of the conventional decoders in coding theory~\cite{T.Gruber2017,Y.Choukroun2022,Y.Choukroun2024b}.
Many of the common architectures have been studied for decoding,
including   multi-layer perceptrons \cite{tseng1993,T.Gruber2017,cammerer2017scaling}, convolutional NNs (CNNs) \cite{sagar1994}, recurrent neural networks (RNNs) \cite{H.Kim2018},
autoencoders \cite{jiang2019deepturbo,clausius2023component,jamali2022productae}, convolutional decoders \cite{G.Marcone95}, graph NNs \cite{cammerer2022graph,clausius2024graph}, and
transformers \cite{Y.Choukroun2024a,Y.Choukroun2022}.
These models have been used to decode   linear block codes \cite{W.R.Caid90,T.Gruber2017},
Reed--Solomon codes \cite{hussain1991reed},
convolutional codes \cite{W.R.Caid90, alston1990, sagar1994},
Bose--Chaudhuri--Hocquenghem (BCH) codes \cite{E.Nachmani18,E.Nachmani2016,M.Lian2019,E.Nachmani2022},
Reed--Muller codes \cite{tseng1993, A.Buchberger2020},  
turbo codes \cite{H.Kim2018}, low-density parity-check (LDPC) codes  \cite{li2023bottom,X.Wu2018}, and polar codes  \cite{Miloslavskaya2024}.

Training neural decoders is challenging because the number of codewords to classify depends exponentially
on the number of information bits.
Furthermore, the sample complexity of the NN is high for the small bit error rates (BERs) and long block lengths required in
some applications.
As a  consequence,   model-agnostic decoders often require a large number of parameters and may overfit,
which makes them impractical unless the block length is  short.

In model-based neural decoders, the architecture of the NN is based on the structure of a conventional decoder
\cite{E.Nachmani2016,vasic2018learning,doan2019neural,A.Buchberger2020,L.Wang2021,E.Nachmani2022}.
An example is  weighted belief propagation (WBP), where the messages exchanged across the edges of the Tanner graph of the code are
weighted and optimized \cite{E.Nachmani2016,E.Nachmani2022,C.Yang2023}.
This gives rise to a decoder in the form of a recurrent network obtained by unfolding the update equations of the belief propagation (BP)  over the iterations.
Since the WBP is a biased model, it has fewer parameters than the model-agnostic NNs at the same accuracy.

Prior work has demonstrated that the WBP outperforms BP for block lengths up to around 1000,
particularly with structured codes, low-to-moderate code rates,
and high signal-to-noise ratios (SNRs) \cite{E.Nachmani2022,Lugosch18,L.Wang2021,Q.Wang2023,C.Yang2023,li2023bottom,M.Wang2023,Q.Wang2023,wang2023ldpc,raviv2024crc}.
It is believed that the improvement is achieved by altering the log-likelihood ratios (LLRs)  that are passed along short cycles.
For example, for BCH and LDPC codes with block lengths under 200,
WBP provides frame error rate (FER) improvements of up to 0.4 dB in the waterfall region and up to 1.5 dB in the error-floor region \cite{Tomer2020,Be2020,kwak2023boosting}.
Protograph-based (PB) QC-LDPC codes have been similarly decoded using the learned
weighted min-sum (WMS) decoder \cite{L.Wang2021}.

The WBP does not generalize well at low bit error rates (BERs) due to the requirement of long block lengths and the resulting
high sample and training complexity \cite{S.Adiga2024}.
For example,  in optical fiber communication, the block length can be up to tens of thousands to achieve a BER of $10^{-15}$.
In this case, the sample complexity of WBP is high, and the model does not generalize well when trained with a practically manageable number of examples.

The training complexity and storage requirements of the WBP can be reduced through parameter sharing.
Lian et al. introduced a WBP decoder wherein the parameters are shared across or within the layers of the NN  \cite{M.Lian2019,wang2023ldpc}.
A number of parameter-sharing schemes in WBP are studied in \cite{L.Wang2021,wang2023ldpc}.
Despite intensive research in recent years, WBP remains impractical in most real-world applications.

In this work, we improve the generalization of WBP to enhance its practical applicability.
The WBP is a static NN, trained offline based on a  dataset.
The main contribution of this paper is the proposal of adaptive learned message-passing algorithms, where the weights assigned to messages are determined for each received word.
In this case, the decoder is dynamic, changing its parameters for each transmission in real time.

Two variants of this decoder are proposed. In the parallel decoder architecture, the weights take values in a discrete set.
A number of WMS decoders are run in parallel to find the best sequence of weights based on the Hamming weight
of the syndrome of the received word.
In the two-stage decoder, a secondary NN is trained to compute the weights to be used in the primary NN decoder.
The secondary NN is a CNN that takes the LLRs of the received word and is optimized offline.

The performance and computational complexity of the static and adaptive decoders are compared in two applications.
In the first application, a number of regular and irregular quasi-cyclic low-density parity-check (QC-LDPC) codes,
along with a BCH  and a polar code, are evaluated over an additive white Gaussian noise (AWGN) channel in both low- and high-rate regimes.
The results indicate that the adaptive WMS decoders achieve decoding BERs up to an order of magnitude less than the BERs of the static WMS
decoders, at about the same decoding complexity, depending on the code, its rate, and the SNR.
The coding gain is 0.32 dB  at a bit error rate of $10^{-4}$ in one example.

The second application is  coding over a nonlinear optical fiber link with wavelength division multiplexing (WDM).
The  data rates in today's optical fiber communication system approach terabits/s per wavelength. Here, the complexity, power consumption, and latency of the decoder are important considerations.
We apply concatenated coding by combining a low-complexity short-block-length soft-decision inner code with
a long-block-length hard-decision outer code. This approach allows the component codes to have much shorter block lengths
and higher BERs than the combined code. As a result, it becomes feasible to train the WBP for decoding the inner code,
addressing the curse of dimensionality and sample complexity issues.
For PB QC-LDPC inner codes and a spatially coupled (SC) QC-LDPC outer code, the results indicate that the adaptive
WBP outperforms the static WBP  by 0.8 dB at about the same complexity and decoding latency in a 16-QAM 8x80 km 32 GBaud WDM system with five channels.

The remainder of this paper is organized as follows.
Section \ref{sec:Notations} introduces the notation, followed by the channel models in Section \ref{sec:model}.
In Section \ref{sec:wbp}, we introduce the WBP, and in Section \ref{sec:adaptive}, two adaptive learned message-passing
algorithms. In Section \ref{sec:results}, we compare the performance and complexity of the static and adaptive decoders,
and in Section \ref{sec:conclusions}, we conclude the paper.
Appendices~\ref{subsec:qc-ldpc} and \ref{sec:sc-code} provide
supplementary information, 
and Appendix~\ref{qc-ldpc-prmts} presents the parameters of the codes.

\section{Notation} \label{sec:Notations}

Natural, real, complex and non-negative numbers are denoted by $\mathbb{N}$, $\mathbb{R}$,  $\mathbb{C}$, and $\mathbb R_+$,  respectively.
The set of integers from $m$ to $n$ is shown as $\frommton{m}{n}=\bigl\{ m, m+1, \cdots, n\bigr\} $.
The special case $m=1$ is shortened to $\fromoneton{n} \eqdef \frommton{1}{n}$. $\lfloor x \rfloor$ and $\lceil x \rceil$ denote, respectively, the floor and ceiling of $x\in\mathbb R$.
The Galois field GF(q) with $q\in \mathbb{N}$, $q\geq 2$, is $\mathbb{F}_q$.
The set of matrices with $m$ rows, $n$ columns, and elements in $\frommton{0}{q-1}$ is $\mathbb{F}^{m\times n}_q$.

A sequence of length $n$ is denoted as $x^n =(x_1, x_2, \cdots, x_n)$.
Deterministic vectors are denoted by boldface font, \eg\ $\vc{x}\in\mathbb{R}^n$.
The $i^{\text{th}}$ entry of $\vc x$ 
is $[\vc x]_i$.
Deterministic matrices are shown by upper-case letters with mathrm font, \eg\ $\matdet{A}$.

The probability density function (PDF) of a random variable $X$ is denoted by ${\Pr}_X(x)$, shortened to $\Pr(x)$
if there is no ambiguity.
The conditional PDF of $Y$ given $X$ is ${\Pr}_{Y|X}(y|x)$.
The expected value of a random variable $X$ is denoted by $\E(X)$.
The real Gaussian PDF with mean $\theta$ and standard deviation $\sigma$ is denoted by $N(\theta,{\sigma}^2)$.
The $Q$ function is $Q\left(x\right) = \frac{1}{2} \erfc\left(\frac{x}{\sqrt{2}}\right)$, where
$\erfc(x)$ is the complementary error function.
The binary entropy function is
$H_b(x) = -(x\log_2(x) + (1-x)\log_2(1-x))$, $x\in (0, 1)$.

\section{Channel Models}
\label{sec:model}

\subsection{AWGN  Channel}
\label{ChanModl_AWGN_En_De}

\emphsec{Encoder}
We consider an $(n, k)$ binary linear code $\mathcal{C}$ with the parity-check matrix (PCM) $\matdet{H} \in \mathbb{F}_2^{m \times n}$,
where $n$ is the code length, $k$ is the code dimension, and $m \geq n - k$, $m\in\mathbb N$.
The rate of the code is $r = k/n \geq 1 - \frac{m}{n}$.
A PB QC-LDPC code is characterized by a lifting factor $M\in\mathbb{N}$, a base matrix $\matdet{B}\in \{-1, 0, 1\}^{\lambda \times \omega}$ , and  an exponent
matrix $\matdet{P} \in \{0, 1, \cdots, M-1\}^{\lambda \times \omega}$ where, $\lambda, \omega\in\mathbb{N}$,
$\lambda<\omega$.
Given $(\lambda, \omega, M, \matdet{P})$, the PCM is obtained according to the procedure in Appendix~\ref{subsec:qc-ldpc}.

We evaluate a BCH code, seven regular and irregular QC-LDPC codes, and  a polar code in the low- and high-rate regimes.
These codes are summarized in Table~\ref{tab:codes} and described in Section~\ref{sec:AWGN-results}.
The parameters of the QC-LDPC codes are given in Appendix~\ref{qc-ldpc-prmts}.

\begin{table*}[t]
\caption{Codes in this paper.}
\label{tab:codes}
\centering
\begin{tabularx}{0.725\textwidth}{CC}
\toprule
\multicolumn{2}{c}{\textbf{AWGN Channel}} \\
\midrule
Low rate & High rate \\ \midrule
BCH $C_1(63, 36)$, $r=0.57$                        & QC-LDPC $\mathcal{C}_6(1050, 875)$, $0.83$         \\\midrule
QC-LDPC $C_2 (3224, 1612)$, $0.5$                & QC-LDPC $\mathcal{C}_7(1050, 850)$, $0.81$         \\\midrule
QC-LDPC $C_3 (4016, 2761)$, $0.69$                &  QC-LDPC $\mathcal{C}_8(4260, 3834)$, $0.9$           \\\midrule
Irregular LDPC $C_4(420, 180)$, $0.43$            &   Polar $\mathcal{C}_9(1024, 854)$, $0.83$     \\\midrule
Irregular LDPC $C_5(128, 64)$, $0.5$             &                                             \\
\midrule
\multicolumn{2}{c}{\textbf{Optical Fiber Channel}} \\
\midrule
Inner code & Outer code \\ \midrule
Single-edge QC-LDPC $\mathcal{C}_{10}(4000,3680)$, $0.92$ &  Multi-edge QC-LDPC $\mathcal{C}_{11}(3680,3520)$, $0.96$   \\\midrule
Non-binary multi-edge $\mathcal{C}_{12}(800, 32)$ &  \\\midrule
\end{tabularx}
\end{table*}

The encoder maps a sequence of information bits $\mathbf{b} = (b_{1}, b_{2}, \ldots, b_{k})$,  $b_{i}\in\{0, 1\}$,
$i\in \fromoneton{k}$, to a codeword $\mathbf{c} = (c_{1}, c_{2}, \ldots, c_{n})$
as $\vc{c} = \vc{b} \matdet{G}$, where $\matdet{G} \in \FF^{k\times n}$ is the generator matrix
of the code.

\emphsec{Channel model}~
The codeword $\mathbf{c}$ is modulated with a binary phase shift keying  with symbols $\pm A\in\mathbb R$,
and transmitted over an AWGN channel. The vector of received symbols is  $\vc{y} = (y_1, \cdots, y_n)$, where
\begin{IEEEeqnarray}{rCl}
    y_i = (-1)^{c_i} A + z_i,\quad i=1,2, 3, \cdots,
    \label{eq:awgn-channel}
\end{IEEEeqnarray}
where
$z_j \sim\iid N(0, \sigma^2)$.
If $\rho$ is the SNR, the channel can be normalized so that $A=1$, and $\sigma^2=(r \rho)^{-1}$.

The LLR function $L:\mathbb R^{n}\mapsto \mathbb R^n$ of $\vc c$ conditioned on $\vc y$ is
\begin{IEEEeqnarray}{rCl}
\left[ \matdet{L}\left(\vc y\right) \right]_i& \eqdef&
 \log\left(\frac{\Pr\left( c_i=0\mid \vc y\right)}{\Pr\left( c_i=1\mid \vc y\right)} \right)   \nonumber \\
   &=&\log\left(\frac{{\Pr}\left(y_i\mid c_i=0\right)}{{\Pr}\left(y_i\mid c_i=1\right)} \right) \label{eq:LLR-awgn-a}\\
   &=& 4r\rho y_i  \label{eq:LLR-awgn}
\end{IEEEeqnarray}
for each $i\in [n]$. Equation  \eqref{eq:LLR-awgn-a} holds under the assumption that $c_i$ are independent and uniformly distributed,
and  \eqref{eq:LLR-awgn} is obtained from Gaussian $\Pr(y_i|c_i)$ from \eqref{eq:awgn-channel}.

\emphsec{Decoder}~
We compare the performance and complexity of the static and adaptive belief propagation.
The static decoders  are tanh-based BP, the auto-regressive BP and WBP with different levels of parameter sharing,
including BP with simple scaling and parameter-adapter networks (SS-PAN) \cite{M.Lian2019}.
Additionally, to assess the achievable performance with a large number of parameters in the decoder,
we include a comparison with two model-agnostic neural decoders based on transformers \cite{Y.Choukroun2024a} and graph NNs  \cite{cammerer2022graph,clausius2024graph}.

\subsection{Optical Fiber Channel}
\label{sec:fiber-channel-model}

In this application, we consider a multi-user fiber-optic transmission system using WDM with $N_c$ users, each of
bandwidth $B_0$ Hz, as shown in Figure~\ref{fig:sys-model}.

\begin{figure*}[t]
\centering
\includegraphics[width=0.982\textwidth]{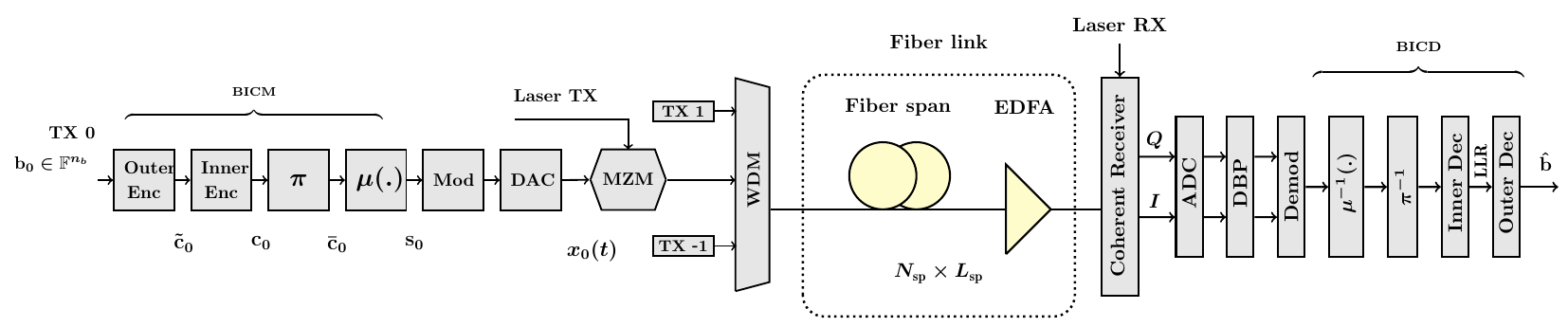}
\caption{Block diagram of an optical fiber transmission system.}
\label{fig:sys-model}
\end{figure*}

\emphsec{Transmitter (TX)}
A binary source generates a pseudo-random bit sequence  of $n_b$ information bits
$\mathbf{b}_{u} = (b_{u,1}, b_{u,2}, \ldots, b_{u,n_b})$, $b_{u, j}\in\{0, 1\}$, for the WDM channel $u\in [u_1:u_2]$,
$u_1 = -\lfloor N_c/2 \rfloor$, $u_2 = \lceil N_c/2 \rceil -1$, $j \in [1:n_b]$.
Bit-interleaved coded modulation (BICM) with concatenated coding  is applied in  WDM channels independently.
BICM comprises an outer  encoder for the code
$\mathcal{C}_{\text{o}}(n_{\text{o}},k_{\text{o}})$ with rate $r_{\text{o}}$, an inner encoder for
$\mathcal{C}_i(n_{\text{i}},k_{\text{i}})$ with rate $r_{\text{i}}$, a permuter $\pi$, and a mapper $\mu$, where $n_{\text{o}}/k_{\text{i}}$ is assumed to be an integer.
The concatenated code  $\mathcal C =\mathcal{C}_i \circ \mathcal{C}_o $ has parameters
$k = k_{\text{o}}$ and $n=n_{\text{o}}/r_{\text{i}}$ and $r_{\text{total}} = r_{\text{i}}r_{\text{o}}$.
Each consecutive subsequence of $\mathbf{b}_u$ of length $k_o$
is mapped to $\tilde{\mathbf{c}}_u\in \mathcal{C}_o\subset\{0,1\}^{n_{\text{o}}} $ by the outer encoder
and subsequently to $\mathbf{c}_u\in\mathcal{C} \subset \{0,1\}^{n}$ by the inner encoder.
Next, $\mathbf{c}_u$ is mapped to $\bar{\mathbf{c}}_u = \pi(\mathbf{c}_u)$ by a random uniform permuter $\pi: \FF^n\mapsto \FF^n $.
The mapper $\mu:\FF^m\mapsto \mathcal A$ maps consecutive sub-sequences of  $\bar{\mathbf{c}}_u$ of length $m$
to a symbol in a constellation $\mathcal A$ of size $\mathfrak{M}=2^m$.
Thus, the BICM maps $\mathbf{b}_u$ to a sequence of symbols $\mathbf{s}_{u} = (s_{u, {\ell}_1}, s_{u,{\ell}_1 + 1}, \ldots, s_{u, {\ell}_2})$,
where $s_{u,{\ell}}\in \mathcal A$, $\ell\in[\ell_1:\ell_2]$, $\ell_1 = -\lfloor n_s/2 \rfloor$, $\ell_2 = \lceil n_s/2 \rceil -1$, $n_s=n_b/m$.

The symbols $s_{u, {\ell}}$ are modulated with a root raised cosine (RRC) pulse shape $p(t)$ at symbol rate $R_s$, where $t$ is the time.
The resulting electrical signal of each channel $x_u(t)$ is converted to an optical signal and subsequently multiplexed by a WDM multiplexer. The baseband representation of the transmitted signal is
\begin{equation}
x(t) = \sum_{u=u_1}^{u_2} \sum_{{\ell}=\ell_1}^{\ell_2} s_{u, {\ell}} \: p(t - {\ell} T_s)e^{j2\pi u B_0 t},
\end{equation}
where $T_s=1/R_s$ and $s_{u, {\ell}}$ are \iid\ random variables.
The average power of the transmitted signal is $\mathcal P$; thus,  $\mathbb E|s_{u, {\ell}}|^2=\mathcal P$, $\forall u, \ell$.

\emph{Encoder:} SC LDPC codes are attractive options for optical communications \cite{Schmalen2015}.
These codes approach the capacity of the canonical communication channels \cite{ShK13,Liga17}
and have a flexible performance--complexity trade-off. They are decoded with the BP and the sliding window decoder (SWD).
Another class of codes in optical communication is the SC product-like codes, braided block codes \cite{feltstrom2009braided}, SC turbo product codes \cite{montorsi2021design}, staircase codes \cite{smith2011staircase,zhang2014staircase,zhang2017analysis} and their generalizations
\cite{shehadeh2023generalized,sukmadji2022zipper}. These codes are decoded with iterative, algebraic hard decision  algorithms and prioritize low-complexity, hardware-friendly decoding over coding gain.

In this paper, the encoding in BICM combines an inner (binary or non-binary) QC-LDPC code $\mathcal C_i$ with an outer SC QC-LDPC code $\mathcal C_o$
whose component code is a multi-edge QC-LDPC code, as outlined in Table \ref{tab:codes}.
The construction and parameters of the codes are given in Appendix \ref{Outer_code} and Appendix \ref{qc-ldpc-prmts}, respectively.

The choice of the inner code is due to the decoder complexity.
Other options have been considered in the literature, for instance,
algebraic codes, \eg\ the BCH  (Sec. 3.3~\cite{ryan2009}) or Reed-Solomon codes (Sec. 3.4~\cite{ryan2009}), or polar codes \cite{Ahmad2016}.
However, the QC-LDPC codes are simpler to decode, especially at high rates.
The outer code can  be an LDPC code  \cite{barakatain2018, Amat19}, a staircase code \cite{Zhang_Kschischang17,barakatain2018},
or a SC-LDPC code \cite{Amat19}.

\emphsec{Fiber-optic link}~
The channel is an optical fiber link  with $N_{sp}$ spans of length $L_{sp}$ of the standard single-mode fiber, with
parameters in Table~\ref{tab:SystemPara}.

\begin{table*}[t]
\caption{The parameters of the fiber-optic link.}
\centering
    \begin{tabularx}{0.45\textwidth}{CC}
        \toprule
        \textbf{Parameter Name}&\textbf{Value} \\
        \midrule 
        \multicolumn{2}{c}{Transmitter parameters} \\
        \midrule
       \text{WDM channels} & 5 \\
        \midrule
        \text{Symbol rate $R_s$} & 32 \text{Gbaud}\\
        \midrule
       \text{RRC roll-off} & 0.01 \\
        \midrule
       \text{Channel frequency spacing} & 33 \text{GHz} \\
        \midrule 
        \multicolumn{2}{c}{Fiber channel parameters}
 \\
        \midrule
       \text{Attenuation} $\left(\alpha \right)$ & 0.2 \text{dB/km} \\
        \midrule
       \text{Dispersion parameter (D)} & 17 \text{ps/nm/km}\\
        \midrule
       \text{Nonlinearity parameter} $(\gamma)$ & 1.2 \text{l/(W·km)}\\
        \midrule
       \text{Span configuration} & 8~$\times$~80 \text{km}\\
        \midrule
       \text{EDFA gain} & 16 \text{dB} \\
        \midrule
        \text{EDFA noise figure } & 5 \text{dB}\\  \bottomrule
    \end{tabularx}

    \label{tab:SystemPara}
\end{table*}

Let $q(t, z):\mathbb R\times\mathbb R^+\mapsto\mathbb C$ be the complex envelope of the  signal
as a function of time $t$ and  distance $z$ along the fiber.
The propagation of the signal in one polarization over one span of optical fiber is modeled by the nonlinear Schr\"odinger
equation \cite{agrawal2019}
\begin{IEEEeqnarray}{rCl}
    \frac{\partial q(t, z)}{\partial z} {=}
-\frac{\alpha}{2}q(t, z)
-\frac{j\beta_2}{2}\frac{\partial^2 q(t, z)}{\partial t^2} 
%\nonumber \\&&
+j\gamma|q(t, z)|^2 q(t, z),
\IEEEeqnarraynumspace
\label{eq:nls}
\end{IEEEeqnarray}
where $ \alpha$ is the loss constant, $\beta_2$ is the chromatic dispersion coefficient, $\gamma$ is the Kerr nonlinearity parameter, and $j=\sqrt -1$.
The transmitter is located at $z = 0$ and the receiver at $z = \mathcal{L}$.
The continuous-time model \eqref{eq:nls} can be discretized to a discrete-time  discrete-space model using the split-step Fourier method (Sec.~III.B~\cite{kramer2015upper}).
The optical fiber channel described by the partial differential equation \eqref{eq:nls} differs significantly from the AWGN channel
due to the presence of nonlinearity.

An erbium doped fiber amplifier (EDFA) is placed at the end of each span, which compensates for the fiber loss, and introduces amplified spontaneous emission noise.
The input $x_{i}(t)$--output $x_o(t)$ relation of the EDFA is given by $x_o(t) = G x_i(t) + n(t)$,
where $G=e^{\alpha \mathcal{L}_{sp}}$ is the amplifier's gain, and $n(t)$ is zero-mean circularly symmetric complex Gaussian noise process
with the power spectral density
\begin{IEEEeqnarray*}{rCl}
\sigma^2 = \frac{1}{2}(G-1)hf_0\textnormal{NF},
\end{IEEEeqnarray*}
where NF is the noise figure, $h$ is a Planck constant, and $f_0$ is the carrier frequency at 1550~nm.

\emphsec{Receiver}
The advent of the coherent detection paved the way for the compensation of transmission effects in optical fiber using
digital signal processing (DSP). As a result, the linear effects in the channel, such as the chromatic dispersion and polarization-induced impairments,
and some of the nonlinear effects, can be compensated with DSP.

At the receiver, a demultiplexer filters the signal of each WDM channel. The optical signal for each channel is converted to an electrical signal by a coherent receiver.
Next, DSP followed by bit-interleaved coded demodulation (BICD) is applied.
The continuous-time electrical signal is converted to the discrete-time signals by analogue-to-digital converters,
down-sampled, and passed to a digital signal processing unit for the mitigation of the channel impairments.
For equalization,   digital back-propagation (DBP) based on the symmetric split-step Fourier method is applied to compensate
for most of the linear and nonlinear fiber impairments \cite{secondini2016}.

After DSP, the symbols  are still subject to   signal-dependent noise, which is mitigated by the bit-interleaved coded demodulator (BICD).
Let $\vc{y}\in\mathbb{C}^{n_s}$ denote the equalized signal samples for the transmitted symbols $\vc{s}\in\mathcal{A}^{n_s}$  in the WDM channel of interest.
Given that the deterministic effects were equalized, we assume that the channel $\vc s\mapsto \vc y$ is memoryless
so that $\Pr(\vc y| \vc s) = \prod_{\ell=1}^{n_s} \Pr(y_{\ell}| s_{\ell})$.
For $s\in\mathcal A$, let $\mu^{-1}(s) = (b_1(s), \cdots, b_m(s))$.
From the symbol-to-symbol channel $\Pr(y | s)$, $s\in\mathcal A$, $y\in\mathbb{C}$,  we obtain $m$ bit-to-symbol channels
\begin{IEEEeqnarray}{rCl}
{\Pr}_j(y | b)= \sum\limits_{s\in\mathcal A, b_j(s)=b} \Pr( y| s),
\label{eq:per-bit-prob}
\end{IEEEeqnarray}
where $b\in\FF$, and $j \in [m]$.

Let
$\bar{\vc c} = (b_1(s_1),\cdots, b_m(s_1), \cdots, b_1(s_{n_s}), \cdots, b_m(s_{n_s}))$, $n=mn_s$.
The LLR function $L:\mathbb C^{n_s}\mapsto \mathbb R^n$ of $\vc c$ conditioned on $\vc y$ is, for each $i\in [n]$,
\begin{IEEEeqnarray}{rCl}
\left[ \matdet{L}\left(\vc y\right) \right]_i& \eqdef&
 \log\left(\frac{\Pr\left( c_i=0\mid \mathbf{y}\right)}{\Pr\left( c_i=1\mid \mathbf{y}\right)} \right)  \nonumber\\
  &=&\log\left(\frac{\Pr\left( \bar c_{i'}=0\mid \mathbf{y}\right)}{\Pr\left( \bar c_{i'}=1\mid \mathbf{y}\right)} \right)  \nonumber\\
 &=&\log\left(\frac{{\Pr}_j\left(y\mid b=0\right)}{{\Pr}_j\left(y\mid b=1\right)} \right),
 \label{eq:LLR-fiber}
\end{IEEEeqnarray}
where $i'$ is obtained from $i$ according to $\pi$,
$j= i' \mod m$,
and ${\Pr}_j(y | b)$ is defined in \eqref{eq:per-bit-prob}.

\vspace{2pt}

\emph{Decoder:} The decoding of $\mathcal{C}_i \circ \mathcal{C}_o$ consists of two steps. First, $\mathcal{C}_i$ is decoded using
an adaptive WBP in Appendix \ref{sec:sc-code}, 
which takes the soft information $\matdet L(\vc y)\in\mathbb R^n$ and corrects some errors.
Second, $\mathcal{C}_o$ is decoded using the min-sum (MS) decoder with SWD  in Appendix~\ref{sec:sc-code},
which further lowers the BER, and outputs the decoded information bits $\hat{\vc b}$.
The LLRs in the inner decoder are represented with 32 bits, and in the outer decoder are quantized at 4~bits with per-window configuration.

In optical communication, the forward error correction (FEC) overhead  6--25\% is common \cite{ITU-G709}.
Thus, the inner code typically has a high rate of $\geq$0.9 and a block length of several thousands,
achieving a BER of ${10}^{-6}$--${10}^{-2}$.
The outer code has a length of  up to tens of thousands, lowering the BER to an error floor to $\sim$$10^{-15}$.

\subsection{Performance Metrics}\label{metrics}

\emphsec{Q-factor}
The SNR per bit in the optical fiber channel is $E_b/N_o$, where $E_b=\mathcal{P}/m$  is the bit energy, and $N_o=\sigma^2 B N_{sp}$
is the total noise power in the link, where $B =B_0N_c$.
The performance of the uncoded communication system is often measured by the BER.
The $Q$-factor for a given BER is the corresponding SNR in an additive white Gaussian noise channel with
binary phase-shift keying modulation:
\begin{IEEEeqnarray*}{rCl}
\QF = 20\log_{10}\Bigl(\sqrt{2}\erfc^{-1}(2\textnormal{BER})\Bigr),\quad \text{dB}.
\end{IEEEeqnarray*}

\emphsec{Coding gain}
Let  $\text{BER}_i$ and $\QF_i$ (respectively, $\text{BER}_o$ and $\QF_o$) denote the BER and $Q$-factor
at the input (respectively, output) of the decoder.
The coding gain (CG) in dB is the reduction in the $Q$-factor
\begin{IEEEeqnarray*}{rCCCl}
\text{CG} &=& \QF_o - \QF_i
&=&20\log_{10} \text{erfc}^{-1}(2\text{BER}_o) - 20\log_{10} \text{erfc}^{-1}(2\text{BER}_i).\label{eq:cg}
 \end{IEEEeqnarray*}
The corresponding net CG (NCG) is
\begin{IEEEeqnarray}{rCl}
\text{NCG} = \text{CG} + 10\log_{10} r_{\text{total}}.
\label{eq:ncg}
\end{IEEEeqnarray}

\emphsec{Finite block-length NCG}
If $n$ is finite, the rate $r_{\text{total}}$ in \eqref{eq:ncg} may be replaced with
the information rate in the finite block-length regime \cite{Polyanskiy2010}
\begin{IEEEeqnarray*}{rCl}
C_f &\approx& C - \log_2(e)\sqrt{\frac{\ber_i(1-\ber_i)}{n}}Q^{-1}\left(\ber_o\right),
\end{IEEEeqnarray*}
where $Q(x) = \frac{1}{2} \, \text{erfc}\left(\frac{x}{\sqrt{2}}\right)$.

\section{Weighted Belief Propagation}
\label{sec:wbp}

Given a code $\mathcal C$, one can construct a bipartite Tanner graph $T_{\mathcal{C}} = (C, V, E)$,
where $C=\fromoneton{m}$,  $\fromoneton{m} \eqdef {1, 2, \dots, m}$, $V=\fromoneton{n}$,
and $E = \{(c,v) \in C \times V \mid \matdet{H}_{c,v} \neq 0\}$
are, respectively, the set of check nodes, variable nodes and the edges connecting them.
Let $V_c = \{v \in V\mid (c,v) \in E\}$, $C_v = \{c \in C\mid (c,v) \in E\}$,
and $d_{c}$ and $d_{v}$ be  the degree of $c$ and $v$ in $T_{\mathcal{C}}$, respectively.

The WBP is an iterative decoder
based on the exchange of the weighted LLRs between the variable nodes and the check nodes in $T_{\mathcal{C}}$  \cite{mezard2009,E.Nachmani2016}.
Let $\matdet{L}^{(t)}_{c2v}$ denote the extrinsic LLR from the check node $c$ to the variable node  $v$ at iteration $t$.
Define similarly $\matdet{L}^{(t)}_{v2c}$.

The decoder is initialized at $t=1$ with $\matdet{L}^{(0)}_{v2c}=  \matdet L(y_j)$, where $v$ is the $j$-th variable node, and $\matdet L(y_j)$ is obtained from \eqref{eq:LLR-awgn} or \eqref{eq:LLR-fiber}.
For iteration $t\in[T]$, the LLRs are updated in two~steps.

The check node update:
\begin{IEEEeqnarray}{rCl}
\matdet{L}^{(t)}_{c2v}  &\overset{(a)}{=}& 2 \tanh^{-1}\Bigl\{\prod_{v' \in\matdet{V}_c\setminus \{v\}} \tanh\frac{\gamma^{(t)}_{v',c}}{2}
                                       \matdet{L}^{(t-1)}_{v'2c}\Bigr\}
\nonumber\\
&\overset{(b)}{\approx}&
\ubar{\matdet L}_{c,v}
\prod_{v'\in \matdet{V}_c\setminus \{ v\}}
\gamma^{(t)}_{v',c}
\operatorname{sign}\left(\matdet{L}^{(t-1)}_{v'2c}\right),
\IEEEeqnarraynumspace
\label{eq:check-node-update}
\end{IEEEeqnarray}
where \[
\ubar{\matdet L}_{c,v} = \min\limits_{v^{\prime}\in \matdet{V}_c \setminus \{v\}} \Bigl\{
\left \lvert \matdet{L}^{(t-1)}_{v'2c}\right\lvert \Bigr\}.
\]
The equation in $(a)$ represents the update relation in the BP \cite{ryan2009}, where the LLR messages are scaled by non-negative weights $\{\gamma^{(t)}_{v, c}: v\in V,  c\in C_v,  t\in[T]\}$.
Further, $(b)$ is obtained from $(a)$ though an approximation to lower the computational cost.
The WBP and WMS decoders use $(a)$ and $(b)$, respectively.

The variable-node update:
\begin{IEEEeqnarray}{rCl}
\matdet{L}^{(t)}_{v2c} &=& \alpha^{(t)}_{v}  \matdet L\left(\vc y\right) + \sum_{c' \in \matdet{C}_{v} \setminus \{c\}}
\beta^{(t)}_{c',v}\matdet{L}^{(t-1)}_{c'2v}.
\IEEEeqnarraynumspace
\label{eq:var-node-update}
\end{IEEEeqnarray}
This is the update relation in the BP, to which the sets of non-negative weights
$\{\alpha^{(t)}_{v}: v\in V,    t\in[T]\}$
and
$\{\beta^{(t)}_{c, v}: c\in C,  v\in V_c,  t\in[T]\}$
are introduced.

At the end of each iteration $t$, a hard decision is made
\begin{IEEEeqnarray}{rCl}
\bar{y}_j=
\begin{cases}
1, &  \text{if}\;\; \matdet{L}^{(t)}_{v}<0,\\
0, & \text{if}\;\; \matdet{L}^{(t)}_{v} \geq 0,
\end{cases}
\label{SyndromCheckHard_1}
\end{IEEEeqnarray}
where
\begin{IEEEeqnarray}{rCl}
          \matdet{L}^{(t)}_{v}= \matdet L\left(y_j\right)+\sum\limits_{c\in C_v} \matdet{L}^{(t)}_{c2v}.
\label{SyndromCheckHard_2}
\end{IEEEeqnarray}
Let $\bar{\mathbf{y}}=\left(\bar y_1,\bar y_2,\cdots,\bar y_n\right)\in\FF^n$, and let $\vc s=\bar{\mathbf{y}}\matdet{H}^{T}\in\FF^m$ be the syndrome.
The algorithm stops if $\vc s =\vc{0}$ or $t=T$.

The computation in \eqref{eq:check-node-update} and \eqref{eq:var-node-update} can be expressed with an NN.
The Tanner graph  $T_{\mathcal{C}}$ is unrolled over the iterations to obtain a recurrent network with $2T$ layers (see Figure~\ref{fig:RNN}), in which
the weights $\gamma^{(t)}_{v,c}$ and $\beta^{(t)}_{c,v}$ are assigned to the edges of $T_{\mathcal C}$, and the weights $\alpha^{(t)}_v$ to the outputs
\cite{E.Nachmani18}.
The weights are obtained by minimizing a loss function evaluated over a training dataset using the standard optimizers for NNs.

\begin{figure*}[t]
\centering
\includegraphics[width=0.95\textwidth]{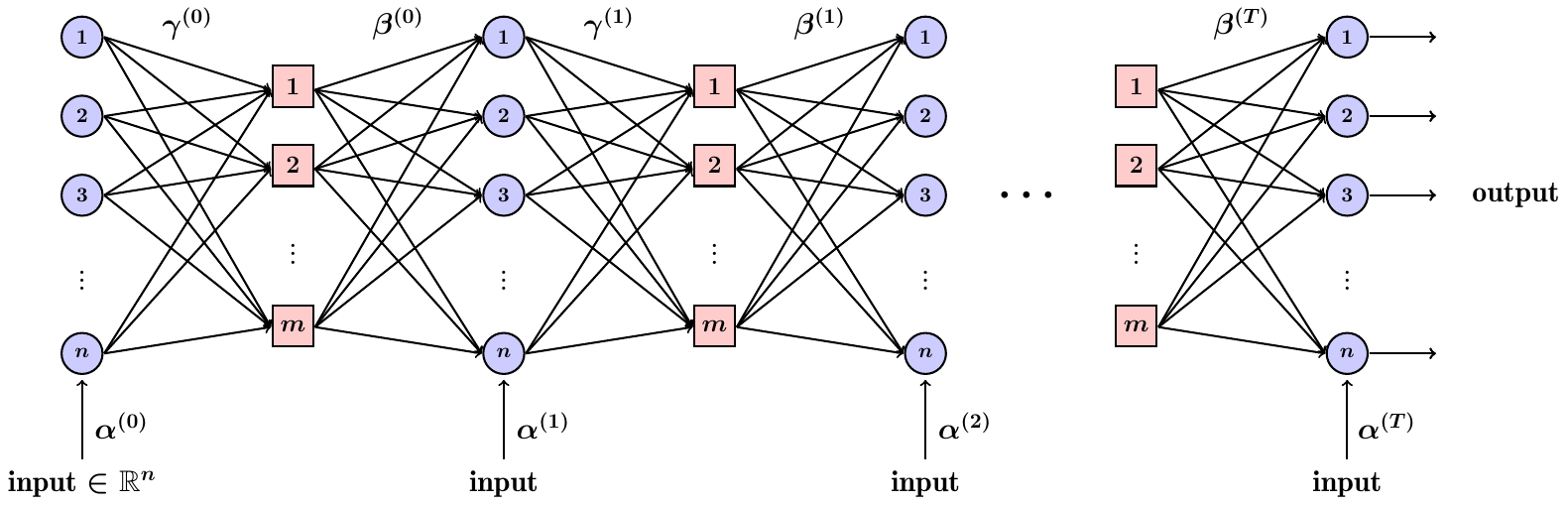}
\caption{Tanner graph unrolled to an RNN.}
\label{fig:RNN}
\end{figure*}

\subsection{Parameter Sharing Schemes}

\label{sec:param-sharing}

The training complexity of WBP can be reduced through parameter sharing at the cost of performance loss.
We consider dimensions $(t, v, c)$ for  the ragged arrays $\gamma_{v,c}^{(t)}$ and $\beta_{c,v}^{(t)}$.
In Type $T$ parameter sharing over $\gamma_{v,c}^{(t)}$,  parameters are shared with respect to iterations $t$.
In Type $T_a$ scheme, $\gamma_{v,c}^{(t)}=\beta_{c, v}^{(t)}=\gamma_{v,c}$,
$\forall\:(c, v)\in E$, $\forall\:t\in[T]$.
In this case, there is a single ragged array with $\vert E\vert$ trainable parameters $\{\gamma_{v,c}\}_{c\in C_v, v \in V_c}$.
For the regular LDPC code, $\vert E\vert = nd_v=md_c$.
It has been observed  that for typical block lengths,
indeed, the weights do not change significantly with iterations \cite{L.Wang2021}.
In Type $T_b$,  there are $T$ arrays $\gamma_{v,c}^{(t)}=\beta_{c, v}^{(t)}$, while in Type $T_c$,
there are two arrays $\gamma_{v,c}^{(t)}=\gamma_{v, c}$ and $\beta_{c, v}^{(t)}=\beta_{c, v}$.
Type $T_a$ and $T_c$ decoders can be referred to as BP-RNN decoders and Type $T_b$ as feedforward BP.
In Type $V$ sharing,  $\gamma_{v,c}^{(t)}=\gamma_{c}^{(t)}$ is independent of $v$. This corresponds to one
weight per  check node.
Likewise, in Type $C$ sharing, there is one weight per  check node update, and $\gamma_{v,c}^{(t)}=\gamma_{v}^{(t)}$.

These schemes can be combined. For instance, in Type $T_aVC$  parameter sharing, $\beta_{c,v}^{(t)}=\gamma_{v,c}^{(t)}=\gamma$.
Thus, a single parameter  $\gamma$ is introduced in all layers of the NN.
This decoder is referred to as the neural normalized BP, \eg\ neural normalized min-sum (NNMS) decoder when BP is based on the MS algorithm.
The latter is similar to the normalized MS decoder, except that the parameter $\gamma$ is empirically determined there.  
In the Type $T_bVC$ scheme, $\beta_{c,v}^{(t)}=\gamma_{v,c}^{(t)}=\gamma^{(t)}$.
Here, there is one weight per iteration. In this paper, $\alpha_v=1$ $\forall\; t\in T$.

\subsection{WBP over $\mathbb{F}_q$}
\label{sec:non-binary}

The construction and decoding of the PB QC-LDPC binary codes can be extended
to codes over a finite field $\mathbb{F}_q$ \cite{Dolecek2014,Boutillon2013}.
Here, there are $q-1$ LLR messages sent from each node, defined in Eq. 1~\cite{Boutillon2013}.
The update equations of the BP are similar to \eqref{eq:check-node-update}--\eqref{eq:var-node-update},
and presented in \cite{Boutillon2013} for the extended min-sum (EMS) and in \cite{Liang2024} for the weighted EMS (WEMS)
decoder.

The parameter sharing for the four-dimensional ragged array $\{\gamma_{v,c, q'}^{(t)}\}_{t, v\in C_v, c \in V_c, q'\in\mathbb{F}_q}$
is defined   in Section~\ref{sec:param-sharing}.  
In the check-node update of the WEMS algorithm, it is possible 
to assign a distinct weight to each coefficient $q' \in \mathbb{F}_q$ for every variable node. For instance, in the Type $T_cCQ$ scheme, 
$\gamma_{v,c,q'}^{(t)}=\gamma_{v}$ and $\beta_{c,v,q'}^{(t)}=\beta_{v}$, so 
there is one weight per variable and one per check node $\forall\; t\in T$. In the case of Type $T_bVCQ$, there is only one weight per 
  variable, iteration, and  coefficient. In this case, if BP is based on the non-binary EMS algorithm, the 
decoder is called the neural normalized EMS (NNEMS).

\begin{remark}
The EMS decoder has a truncation factor in $\{1, 2, \cdots, q\}$ that
provides a trade-off between complexity and accuracy.
In this paper, it is set to $q$ to investigate the maximum performance.
\end{remark}

\section{Adaptive Learned Message Passing Algorithms}
\label{sec:adaptive}

The weights of the static WBP are obtained by training the network offline using a dataset.
A WBP where the weights are determined for each received word  $\vc y$ is an adaptive WBP.
The weights must therefore be found by online optimization.
To manage the complexity, we consider a WMS decoder with Type $T_a$ parameter sharing.
Thus, the decoder has one weight per   $T$ iterations, which must be determined for a received $\vc y$.

Let $\vc c\in\FF^n$ be a codeword and $\mathbf y\in U$ be the corresponding received word, where $U=\mathbb R^n$ 
for the AWGN channel and $U=\mathbb C^{n_s}$ for the optical fiber channel.
Let $\bar{\mathbf y} = \mathcal{D}_{\pmb{\gamma}}(\vc y)$ be the word decoded by a Type $T_a$ WBP decoder 
with weight  $\gamma^{(t)}\in \mathbb R^+$ in iteration $t\in [T]$, where 
$\pmb{ \gamma} = ( \gamma^{(1)}, \gamma^{(2)}, \cdots, \gamma^{(T)})$.
In the adaptive decoder, we wish to find a function 
$g: U\mapsto \mathcal X$, $\mathcal X\subseteq\mathbb R_+^T$, $\pmb{\gamma} = g(\vc y)$ 
that minimizes the probability that $\mathcal{D}_{g(\vc y)}(\vc y)$ makes an~error
\begin{IEEEeqnarray}{rCl}
\min_{g\in\mathcal H} \text{Pr}\bigl(\bar{\mathbf y}\neq  \vc c\bigr), 
\label{eq:gamma-opt}
\end{IEEEeqnarray}
where $\mathcal H $ is a functional class. The static decoder is a special case where  $g(.)$
is a constant function.
Two variants of this decoder are proposed, illustrated in Figure \ref{fig:adaptive-decoders}.

\begin{figure}[t]
\centering 
  \begin{tabular}{c}
      \includegraphics[scale=0.7]{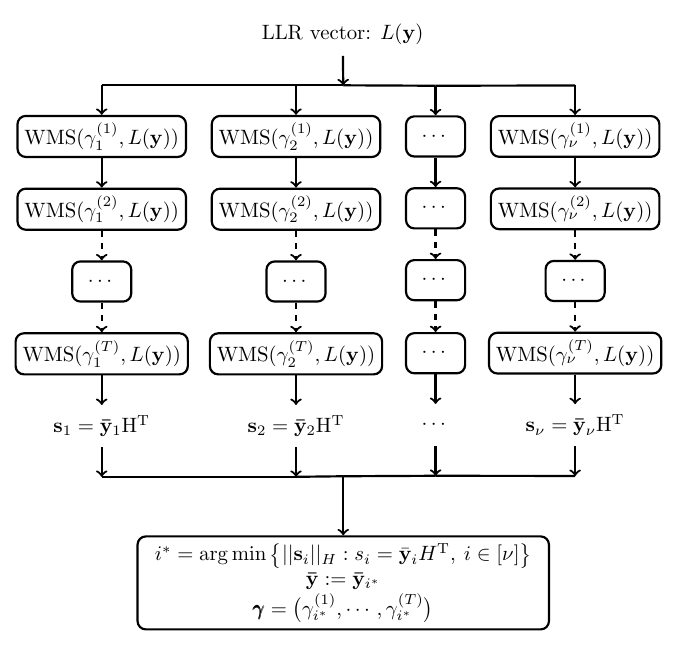} \\[2mm]
      (a) \\[2mm]
      \includegraphics[scale=0.75]{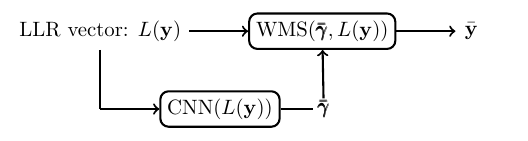}\\[2mm]
      (b)
\end{tabular}
    \caption{Adaptive decoders. (a) Parallel decoders; (b) the two-stage decoder. $\text{WMS}(\pmb{\gamma}, L(\vc y))$ refers 
        to $\mathcal{D}_{\pmb{\gamma}}(\vc y)$.
    }
    \label{fig:adaptive-decoders}
\end{figure}

\subsection{Parallel Decoders}

\emphsec{Architecture}
In parallel decoders, $g(.)$ is found through searching. Here, $\gamma^{(t)}$ takes value in a discrete set $\mathcal X_t = \bigl\{x_1^{(t)}, x_2^{(t)}, \cdots, x_{K_t}^{(t)} \bigr\}$,
$K_t \in \mathbb{N}$, and thus  
$\pmb{\gamma}\in \mathcal X :=\prod_{t=1}^T \mathcal{X}_t = \bigl\{ \pmb{\gamma}_1, \cdots, \pmb{\gamma}_{\nu} \bigr\}$.
The parallel decoders consist of $\nu$ independent decoders $\bar{\vc y}_i= \mathcal D_{\pmb{\gamma}_{i}}(\vc y)$, $i\in[\nu]$, running concurrently.
Since $\text{Pr}\bigl(\bar{\mathbf y}_i\neq  \vc c\bigr)$ in \eqref{eq:gamma-opt} is generally intractable, a sub-optimal $g(.)$ is selected as follows. 
At the end of decoding by $\mathcal D_{\pmb{\gamma}_{i}}$ , the syndrome $\vc s_i= \bar{\mathbf y}_iH^T$ is computed. 
Let
\begin{IEEEeqnarray}{rCl}
i^* = \argmin_{i\in [\nu]} {\left\lbrace ||\mathbf{s}_{i}||_H: \:\:\vc s_i= \bar{\mathbf y}_iH^T \right\rbrace},
\label{eq:i-star}
\end{IEEEeqnarray}
be the index of the decoder whose syndrome  has the smallest Hamming weight.
Then, $g(\vc y) = \pmb{\gamma}_{i^*}$, and $\bar{\mathbf y} = \mathcal{D}_{\pmb{\gamma}_{i^*}}(\vc y)$.

In practice, the search can be performed up to depth $T_1 = 5$ iterations. However, the BP decoder often has to run for more iterations.
Thus, a WBP decoder with weights $\pmb{\gamma}_{i^*}$ can continue the output with $T_2$ iterations.

\begin{remark}
The decoder obtained via  \eqref{eq:i-star} is generally sub-optimal.  
Minimizing $ \norm{\vc s_i}_H$ does not necessarily minimize the number of errors.
However, for random codes, the decoder obtained from \eqref{eq:i-star} outperforms the static decoder.
\end{remark}

\begin{remark}
\label{Remark3}
If $D_{\pmb{\gamma}_1}$ and $D_{\pmb{\gamma}_2}$ yield the same number of errors,
the decoder with the smaller weight vector is selected, which tends to output smaller LLRs.
\end{remark}

\emphsec{Obtaining $x_k^{(t)}$ from the distribution of weights.}
The values of $x_k^{(t)}$ can be determined by dividing a sub-interval in $[0,1]$ uniformly.
The resulting parallel WMS decoder outperforms WMS; however, the performance can be
improved by choosing $x_k^{(t)}$ based on the probability distribution of the weights.

The probability distribution of the channel noise induces a distribution on $\vc y$ and
consequently on $\pmb{\gamma}=g(\vc y)$.
Let $\Gamma^{(t)}$ be a random variable representing
$\gamma^{(t)}$. Denote the corresponding mean by $\theta^{(t)}$, standard deviation by $\sigma^{(t)}$, and the cumulative distribution function by
$C_t (.) \eqdef C_{\Gamma^{(t)}}(.)$.
For $\epsilon>0$, set
\begin{IEEEeqnarray}{rCl}
x_k^{(t)} =\inf\Bigl\{ x: C_t(x) >\frac{\epsilon}{2}+ (k-1)\frac{1-\epsilon}{K_t} \Bigr\}.
\label{eq:xk-1}
\end{IEEEeqnarray}
The numbers $x_1^{(t)}<x_2^{(t)}<\cdots <x_{K_t}^{(t)}$ partition the real line into intervals
such that $\Pr\Bigl(\Gamma^{(t)}\in[x_1^{(t)},x_{K_t}^{(t)}]\Bigr)$ $= 1-\epsilon$ and
$\Pr\Bigl(\Gamma^{(t)}\in[x_k^{(t)},x_{k+1}^{(t)}] \:\Bigl|\: \Gamma^{(t)}\in[x_1^{(t)},x_{K_t}^{(t)}]\Bigr)$ $=\frac{1}{K_t}$.
In practice,  $\Gamma^{(t)}$ has a distribution close to Gaussian, in which case $x_k^{(t)}$ values are given by the explicit formulas in Lemma~\ref{lemm:partition}.

\begin{lemma} \label{lemm:partition}
Let $\Gamma^{(t)}$ have a cumulative distribution  function $C_t(.)$ that is continuous
and strictly monotonic.
For $k\in[K_t]$,
\begin{IEEEeqnarray*}{rCl}
    x^{(t)}_k
    = C_t^{-1}\left(\frac{\epsilon}{2}+ (k-1)\frac{1-\epsilon}{K_t}\right).
\end{IEEEeqnarray*}
If $\Gamma^{(t)}$ has a Gaussian distribution with a mean $\theta^{(t)}$ and standard deviation $\sigma^{(t)}$, then
\begin{IEEEeqnarray}{rCl}
    \label{eq:xk}
    x^{(t)}_k =\theta^{(t)} +\sqrt{2}\sigma^{(t)} \erfc^{-1} \Bigl(2-\epsilon-\frac{2(k-1)(1-\epsilon)}{K_t}\Bigr),
\IEEEeqnarraynumspace
\end{IEEEeqnarray}
where $\erfc$ is the complementary error function.
\end{lemma}
\begin{proof}
The proof is based on elementary calculus.
\end{proof}

\emphsec{Obtaining the distribution of weights}
To apply \eqref{eq:xk-1} or \eqref{eq:xk}, $C_t(.)$ is required.
To this end, a static WBP (with no parameter sharing) is trained offline given a dataset $\{ (\vc{y}^{(i)}, \vc{c}^{(i)})\}_{i}$.
The empirical cumulative distribution of the weights in each iteration is computed as an approximation to $C_t(.)$. 
However, if the BER is low, it can be difficult to obtain a dataset that contains a sufficient number of examples corresponding to 
incorrectly decoded words required to obtain good generalization.

To address this issue, we apply active learning \cite{fu2013survey,Be2020}.
This approach is based on the fact that the training examples near the decision boundary of the optimal classifier determine the classifier
the most. Hence, input examples are sampled from a probability distribution with a support near the decision boundary.

The following approach to active learning is considered. At epoch $\mathord{e}$ in the training of the WBP, random codewords $\vc c$ and
the corresponding outputs $\vc y$ are computed.
The decoder from the epoch $\mathord{e}-1$ is applied to decode $\vc y$ to  $\hat{\vc c}=\text{WBP}(\pmb{\gamma}, L(\vc y))$.
An acquisition function $A_f(\vc c, \hat{\vc c}):\mathcal{C} \times \FF^n \mapsto \mathbb{R}^+$
evaluates whether the example pair $(\vc y, \vc c)$ should
be retained.
A candidate example is retained if $A_f(\vc c, \hat{\vc c})$ is in a given range.

The choice of the acquisition function depends on the specific problem being solved, the architecture of the NN, and
the availability of the labeled data \cite{fu2013survey}.
In the context of training the NN decoders for channel coding, the authors of
\cite{Be2020} use \emph{distance parameters} and \emph{reliability parameters}.
Inspired by \cite{Be2020}, the authors of \cite{H.Noghrei2023} define the
acquisition function using importance sampling.
In this paper, the  acquisition function is the number of errors  $A_f(\vc c, \hat{\vc c}) =\|\vc c -\hat{\vc c}\|_H = d_H (\vc c, \hat{\vc c})$, where $d_H$ is the Hamming distance.

The dataset is incrementally generated and pruned as follows.
At each epoch $e$, a subset $S_e  =  \{ (\vc{y}^{(i)}, \vc{c}^{(i)}) \}_{i=1}^{\mathfrak{b}_1}$ of $\mathfrak{b}_1$ examples, filtered by the acquisition function, is selected.
The entire dataset at epoch $e$ is $S^e = \text{Prune}(\cup_{e'=1}^e S_{e'})$ and has size $\mathfrak{b}_2>\mathfrak{b}_1$.
The operator Prune removes the subsets $S_{e'}$ introduced in old epochs  $e'\in [e-e_0]$ if $e>e_0\eqdef \mathfrak{b}_2/\mathfrak{b}_1$ and otherwise leaves its input intact.
At each epoch $e$, the loss function is averaged over a batch set of size $\mathfrak{b}_s$ obtained by randomly sampling from $S^e$.

\emphsec{Complexity of the parallel decoders}
The computational complexity of the decoder is measured  in real multiplications (RMs).
For instance, the complexity of the WMS with $T$ iterations, $\alpha_v=1$,
without parameter sharing, or with Type $T_b$ parameter sharing, is $\text{RM} = 2T |E|$, where $|E|=nd_v=md_c$ is the number of edges
of the Tanner graph of the code.
For the WMS decoder with Type $T_a$ or Type $T_b V C$ parameter sharing,
$\text{RM} = T\left(m+n\right)$.
The latter arises from the fact that equal weights factor out of the $\sum$ and the $\min$ terms in BP and are applied once.
Thus, the complexity of $\nu$ parallel WMS decoders with Type $T_a$ parameter sharing
is $\text{RM} = \nu T\left(m+n\right)$.
If $\alpha_v\neq 1$, $nT$ is added to the above formulas. Finally, the complexity of Type $T_a V C$ decoder is RM $=n$ per single iteration. These expressions neglect the cost of the syndrome check.

\subsection{Two-Stage Decoder}

In a parallel decoder, the weights are restricted in a discrete set.
The number of parallel decoders depends exponentially on the size of this set.
The two-stage decoder predicts arbitrary non-negative weights, without the exponential complexity of the parallel decoders.
Further, since the weights are arbitrary, the two-stage decoder can improve upon the performance of the parallel decoders, 
when the output LLRs are sensitive to the weights.

\emphsec{Architecture}
Recall that we wish to find a function $g(\vc y)$ that minimizes the BER in \eqref{eq:gamma-opt}.
In a two-stage decoder, this function is expressed by an NN $ \bar{\pmb{\gamma}} = g_{\boldsymbol{\theta}}(\vc y)$ 
parameterized by vector $\boldsymbol{\theta}$. 
Thus, the two-stage decoder is a combination of an NN and a WBP.
First, the NN takes as input  either the LLRs at the channel output $\matdet L(\vc y)$ or $(\matdet L(\vc y), \vc y)$ 
and outputs the vector of weights $\bar{\pmb{\gamma}}$.
Then, the WBP decoder takes the channel LLRs $\matdet L(\vc y)$ and  weights $\bar{\pmb{\gamma}}$ and outputs the decoded word $\bar{\vc y}$.

The parameters $\boldsymbol{\theta}$ are  found using a dataset of examples $\bigl\{(\vc{y}^{(i)}, \pmb{\gamma}^{(i)})\bigr\}_{i}$, 
where $\pmb{\gamma}^{(i)}$ is the target weight.
This dataset can be obtained through a simulation, \ie, transmitting a codeword $\vc{c}^{(i)}$, receiving $\vc{y}^{(i)}$,  
and using, \eg\ an offline parallel decoder to determine the target weight  $\pmb{\gamma}^{(i)}$.
In this manner, $g_{\boldsymbol{\theta}}(\vc y)$ is expressed in a functional form 
instead of being determined by real-time search, which may be more expensive.

In this paper, the NN is a CNN consisting of a cascade of two one-dimensional convolutional layers $\mathbf{Conv}_1$ and $\mathbf{Conv}_2$, followed
by a dense layer $\mathbf{Dens}$.
\(\mathbf{Conv}_i\) applies $F_i$ filters of size $S_i$ and stride 1, and the rectified linear unit (ReLU) activation, $i=1,2$.
The output of \(\mathbf{Conv}_2\) is flattened and passed to \(\mathbf{Dense}\),
which produces the vector of weights $\bar{\pmb{\gamma}}$ of length $T$. This final layer is a linear transformation
with  ReLU activation to produce non-negative weights.

The model is trained by minimizing the quantile loss function
\begin{IEEEeqnarray*}{c}
\ell_{\xi}(\pmb{\gamma}, \bar{\pmb{\gamma}}) = {\rm mean} \Bigl( \max\left(\xi (\pmb{\gamma} - \bar{\pmb{\gamma}}), (\xi - 1)  (\pmb{\gamma} - \bar{\pmb{\gamma}})\right)\Bigr),
\end{IEEEeqnarray*}
where $\xi\in(0,1)$ is the quantile parameter and $\max$ is applied per entry and mean over vector entries.
The choice of loss is obtained by cross-validating the validation error over a number of candidate functions.
This is an asymmetric absolute-like loss, which, if $\xi\geq 1/2$ as in Section \ref{sec:results}, encourages entries of 
$\bar{ \pmb{\gamma}}$ to be close to entries of $\pmb{\gamma}$ from above rather than below.

\emphsec{Complexity of the two-stage decoder}
The computational complexity of the two-stage decoder in the inference mode is the sum of the complexity of the CNN and WMS decoder
\begin{IEEEeqnarray*}{rCl}
\text{RM} = T(m+n) + \text{RM}(\text{CNN}),
\end{IEEEeqnarray*}
where the computational complexity of the CNN is
\begin{IEEEeqnarray}{rCl}
\label{eq:complex_CNN}
\text{RM}(\text{CNN}) &=&  \text{RM}(\mathbf{Conv}_1) + \text{RM}(\mathbf{Conv}_2) + \text{RM}(\mathbf{Dense}) \nonumber
                     \\&=& F_1 (n - S_1 + 1)  S_1 + F_1F_2  (n - S_1 - S_2 + 2) S_2 \nonumber\\
                       &&+ F_2 (n - S_1 -S_2 + 2).
\end{IEEEeqnarray}
The complexity can be significantly reduced by pruning the weights, for example,
by setting to zero the weights below a threshold $\tau_{\text{prun}}$.

\begin{remark}
Neural decoders are sensitive to distribution shifts and often require retraining when the input distribution or channel conditions change. 
To lower the training complexity,  \cite{M.Lian2019} proposed a decoder that learns a mapping from the input SNR to the weights in WBP, 
enabling the decoder to operate across a range of SNRs.
However, the WBP decoder in \cite{M.Lian2019}  is static, since the weights remain fixed throughout the transmission once chosen, 
despite being referred to as dynamic WBP.
We do not  address the problem of distribution shift in this paper.
\end{remark}

\section{Performance and Complexity Comparison}

\label{sec:results}

In this section, we study the performance and complexity trade-off of the static and adaptive decoders
for the AWGN in Section \ref{ChanModl_AWGN_En_De} and optical fiber channel in Section~\ref{sec:fiber-channel-model}.

\subsection{AWGN Channel}
\label{sec:AWGN-results}

\emph{Low-rate regime:}
To investigate the error correction performance of the decoders at low rates,
we consider a BCH code $\mathcal{C}_1(63, 36)$ of rate $0.57$ with the cycle-reduced parity check matrix $\mathrm{H}_{\text{cr}}$ in \cite{HelmMich19}
and two  QC-LDPC codes $\mathcal{C}_2(3224,1612)$ and $\mathcal{C}_3(4016,2761)$, which are, respectively, $(4,8)$- and $(5,16)$-regular
with rates of $0.5$ and $0.69$.
The parity check matrix of each QC-LDPC code is constructed using an exponent matrix obtained
from the random progressive edge growth (PEG) algorithm \cite{Hu2001}, with a girth-search depth of two, which is  subsequently refined manually
to remove the short cycles in their Tanner graphs.
The parameters of the QC-LDPC codes, including the exponent matrices $P_2$ and $P_3$, are given in Appendix \ref{qc-ldpc-prmts}.
In addition, we consider the irregular LDPC codes $\mathcal{C}_4(420,180)$ specified in the 5G New Radio (NR) standard and
$\mathcal{C}_5(128,64)$ in the Consultative Committee for Space Data Systems (CCSDS) standard \cite{HelmMich19}.
The code parameters, such as exponent matrices, are also available in the public repository \cite[v0.1]{repo}.

We compare our adaptive decoders with tanh-based BP, the auto-regressive BP and several static WMS decoder with different levels of parameter sharing,
such as BP with SS-PAN \cite{M.Lian2019}.
The latter is a Type $T_aVC$ WBP with $\alpha_v =\alpha$, \ie, a BP with two parameters.
Additionally, to assess the achievable performance with a large number of parameters in the decoder,
we include a comparison with two model-agnostic neural decoders based on the transformer \cite{Y.Choukroun2024a}
and graph NNs  \cite{cammerer2022graph,clausius2024graph}.

The number of iterations in the WMS decoders of the parallel decoders is chosen so that the total computational complexities of the
parallel decoders  and the static WMS decoder are about the same.
In Figure~\ref{fig:BER_performance}a, this value is $T=5$ for $\mathcal{C}_1$, where $T=T_1+T_2$ is the total number of iterations;
in Figure~\ref{fig:BER_performance}b--d, $T=4$  for $\mathcal{C}_2$ and $\mathcal{C}_3$, $T=6$ for $\mathcal{C}_4$, and $T=10$ for $\mathcal{C}_5$;
in Figure~\ref{fig:BER_performance}f, $T=10$ for $\mathcal{C}_1$.
Furthermore, $K_t=4$ for all $t \in [T]$.

\begin{figure*}[t]
  \centering
    \begin{tabular}{c@{~~}c@{~~}c}
\includegraphics[width=0.32\textwidth]{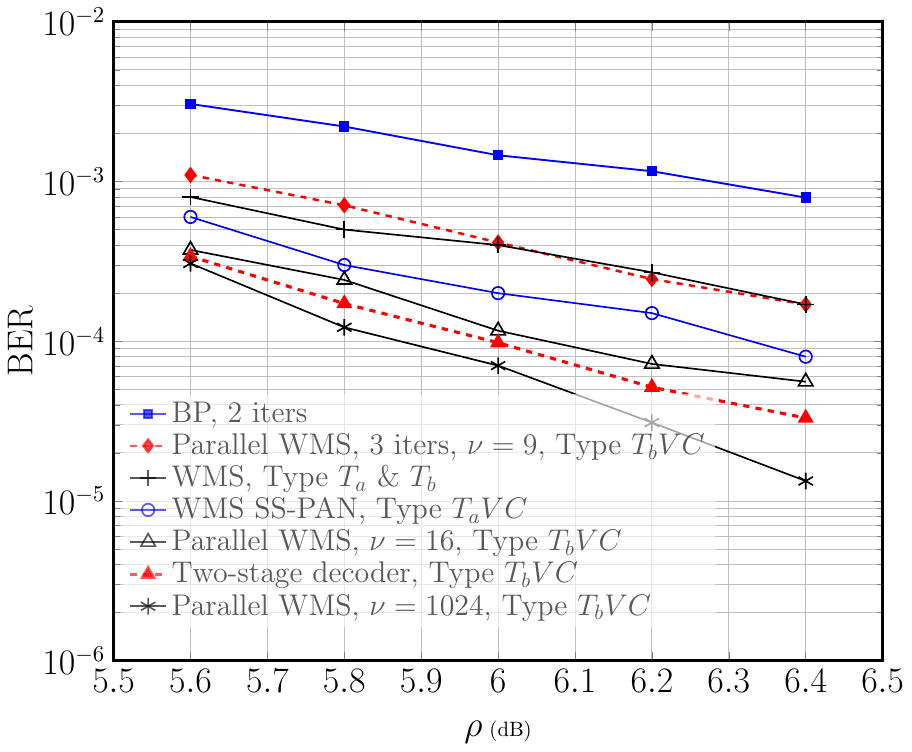} &
\includegraphics[width=0.32\textwidth]{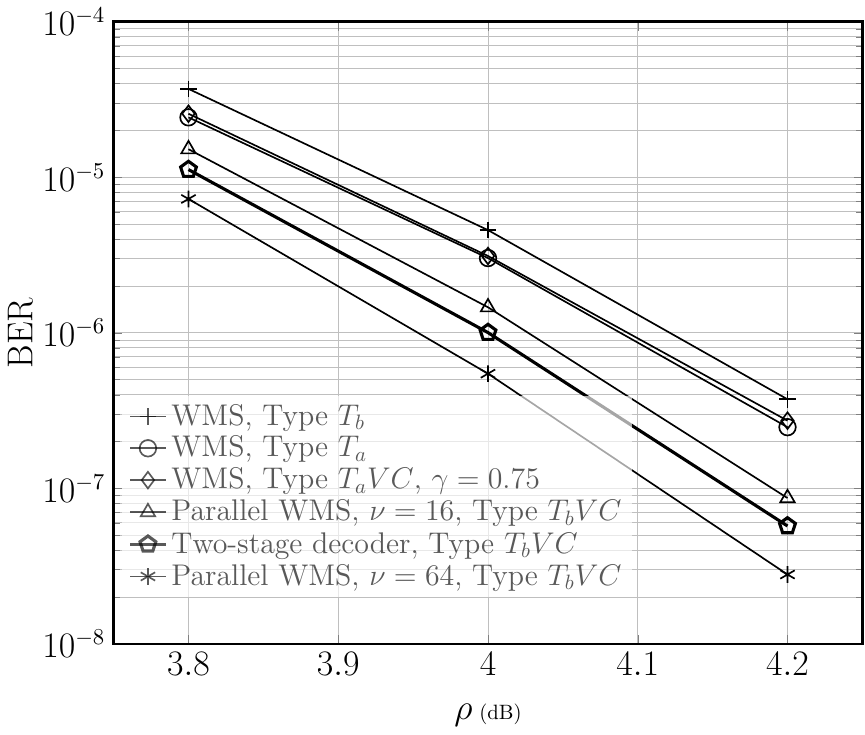} &
\includegraphics[width=0.32\textwidth]{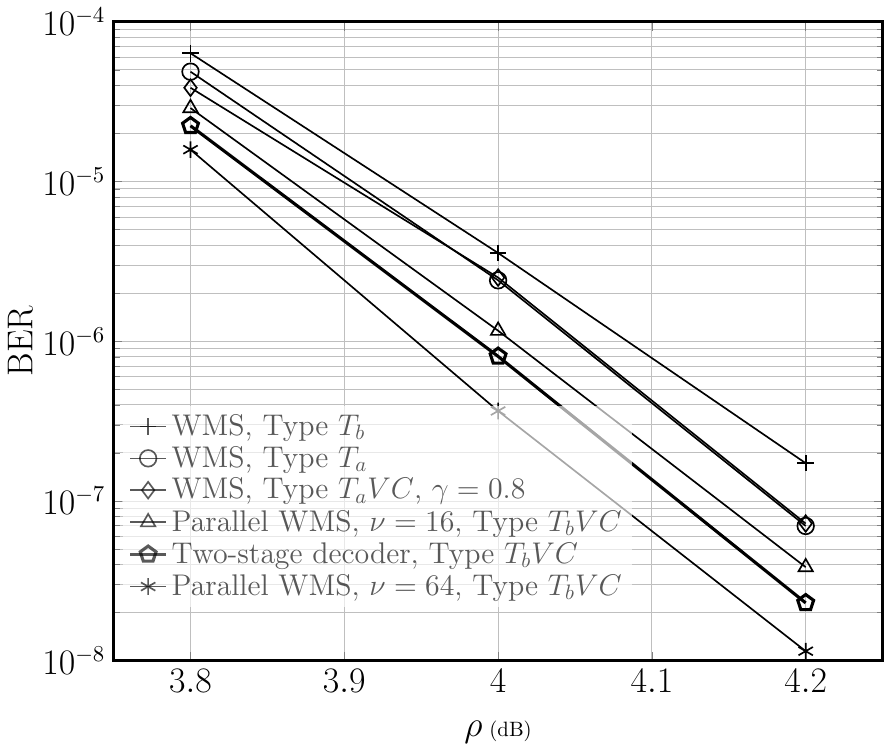} \\
~~(a) & ~~(b) & ~~ (c) \\[2mm]
\includegraphics[width=0.32\textwidth]{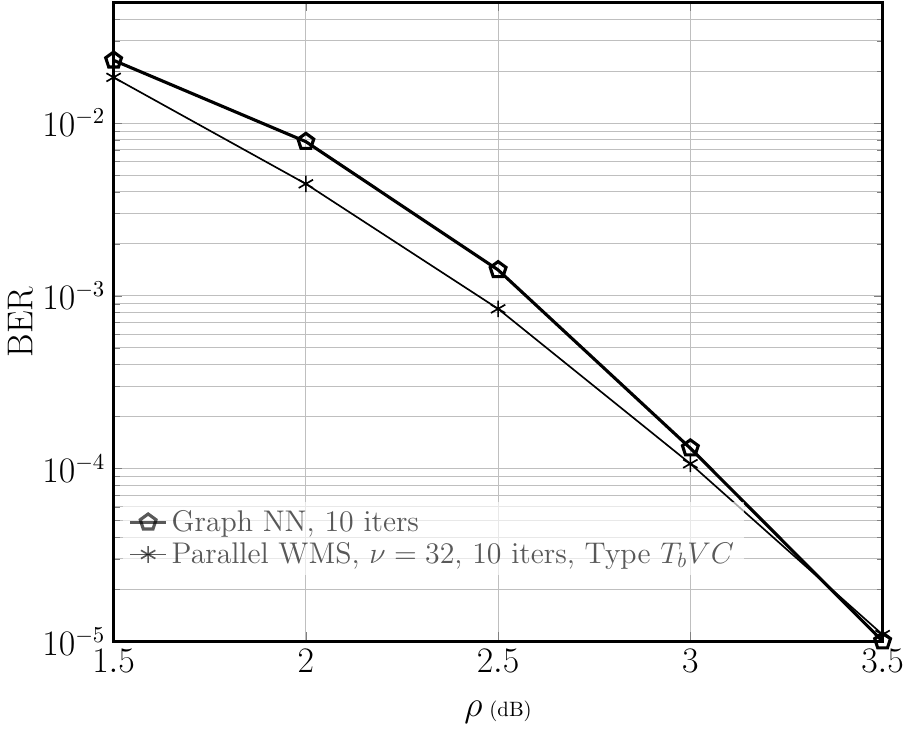} &
\includegraphics[width=0.32\textwidth]{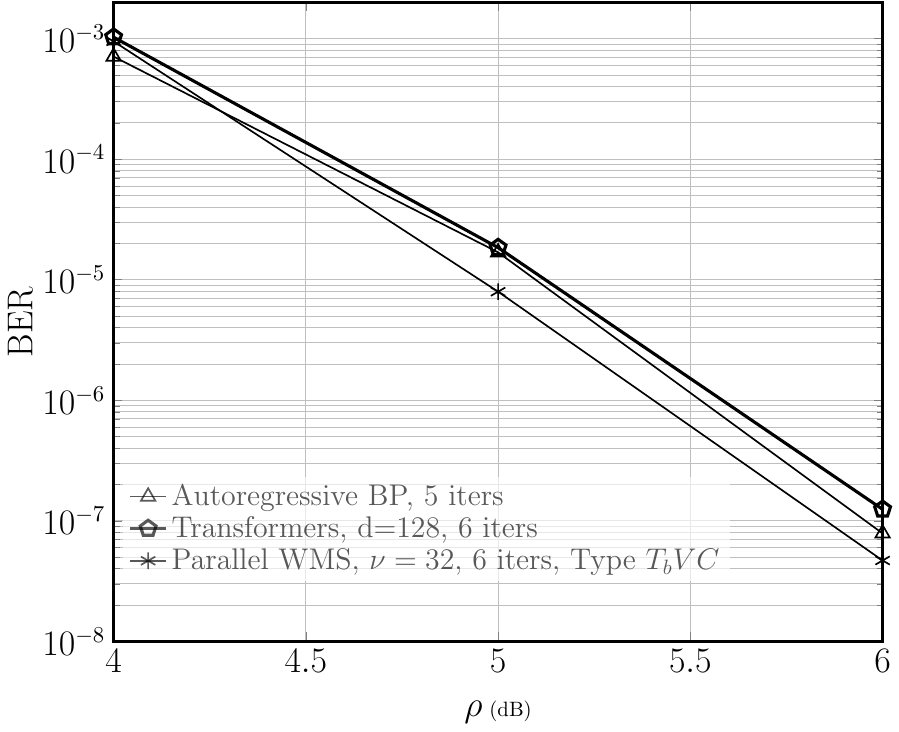} &
\includegraphics[width=0.32\textwidth]{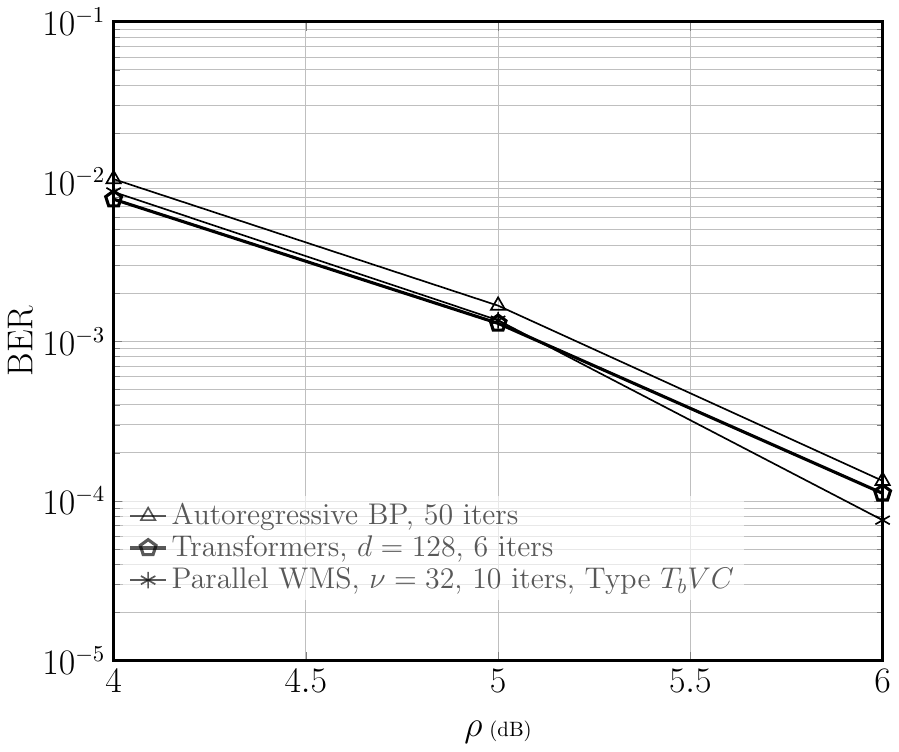} \\
~~~(d) & ~~~(e) & ~~~(f)
  \end{tabular}
  \caption{BER versus SNR $\rho$, for the AWGN channel in the low-rate regime. 
      (a) BCH code $\mathcal C_1(63, 36)$. Here, the curve for WMS, Type $T_a$ \& $T_b$ is from \cite[Fig. 8]{E.Nachmani18} and 
      the curve for WMS SS-PAN is from \cite[Fig. 5(a)]{M.Lian2019}. 
      (b) QC-LDPC code $\mathcal{C}_2(3224, 1612)$,
      (c) QC-LDPC code $\mathcal{C}_3(4016,2761)$, 
      (d) 5G-NR LDPC code $\mathcal{C}_4(420,180)$. Here, the curve for Graph NN is from \cite[Fig. 5]{cammerer2022graph}.
      (e) CCSDS LDPC code $\mathcal{C}_5(128,64)$. In this and the next sub-figure, the Autoregressive BP and Transformers curves are from \cite{nachmani2021autoregressive} and 
      \cite{Y.Choukroun2022}, respectively.
  (\textbf{f}) BCH code $\mathcal C_1(63, 36)$. 
Figures (\textbf{d}--\textbf{f}) show that adaptive decoders achieve the  performance of the static decoders with less complexity.
}
    \label{fig:BER_performance}
\end{figure*}

To compute the value of weights $x_k^{(t)}$, the probability distribution of  $\Gamma^{(t)}$ is required.
For this purpose, a WMS decoder is trained offline.
The training dataset is a collection of examples obtained using
the AWGN channel with a range of SNRs $\rho\in \bigl\{ 5.8, 6.0, 6.2 \bigr\}$ dB
for $\mathcal{C}_1$ or $\rho\in \bigl\{3.8, 3.9, 4.0, 4.1, 4.2 \bigr\}$ dB for $\mathcal{C}_2$ and $\mathcal{C}_3$.
The datasets for  $\mathcal{C}_4$ and $\mathcal{C}_5$ are obtained similarly, with different sets of SNRs.
The acquisition function $A_f(.,.)$ in  active learning
is the Hamming distance.
A candidate example for the training dataset  is retained if  $A_f\leq 10$.
The parameters of the active learning are $b_1=b_s=2000$, and $b_2=40000$.
The loss function is the binary cross-entropy.
The models are trained using the Adam optimizer with a learning rate of $0.0005$.
It is observed that the distribution of $\Gamma^{(t)}$ is nearly Gaussian.
Thus, we obtain $x^{(t)}_k$ from \eqref{eq:xk}.
Table \ref{tab:Gaussian_parameters} presents the mean and variance of this distribution, and $x_k^{(t)}$, for
the three codes considered.
\begin{table*}[t]

    \caption{The mean and variance $(\theta^{(t)}, \sigma^{(t)})$ of $(x_1^{(t)}, \cdots, x_4^{(t)})$ in WMS for the AWGN channel.}
\centering
    \begin{tabularx}{0.64\textwidth}{CCCC}
        \toprule
\textbf{Code}& \boldmath{$\mathcal{C}_1$}
& \boldmath{$\mathcal{C}_2$} &  \boldmath{$\mathcal{C}_3$} \\ 
\midrule
 $t=1$ & \makecell{$\left(0.99,0.019\right)$\\ $\left(0.96,0.98,0.99,1.02\right)$} & \makecell{$\left(0.90,0.026\right)$\\ $\left(0.86,0.89,0.90,0.94\right)$} & \makecell{$\left(0.91,0.023\right)$\\ $\left(0.87,0.90,0.92,0.95\right)$}
 \\ 
\midrule
 $t=2$ & \makecell{$\left(0.97,0.036\right)$ \\ $\left(0.91,0.96,0.98,1.02\right)$} & \makecell{$\left(0.84,0.029\right)$ \\ $\left(0.79,0.83,0.85,0.89\right)$} & \makecell{$\left(0.86,0.030\right)$ \\ $\left(0.81,0.85,0.87,0.90\right)$}
 \\ 
\midrule
 $t=3$ & \makecell{$\left(0.91,0.049\right)$ \\ $\left(0.83,0.89,0.92,0.99\right)$} & \makecell{$\left(0.73,0.032\right)$ \\ $\left(0.68,0.72,0.74,0.78\right)$} & \makecell{$\left(0.75,0.031\right)$ \\ $\left(0.69,0.74,0.76,0.80\right)$}
 \\ 
\midrule
 $t=4$ & \makecell{$\left(0.70,0.086\right)$ \\ $\left(0.56,0.67,0.73,0.84\right)$} & \makecell{$\left(0.63,0.036\right)$ \\ $\left(0.57,0.62,0.64,0.69\right)$}& \makecell{$\left(0.63,0.034\right)$ \\ $\left(0.57,0.61,0.64,0.68\right)$}
 \\ 
\midrule
 $t=5$ & \makecell{$\left(0.40,0.175\right)$ \\ $\left(0.12,0.34,0.46,0.68\right)$} & -- & --
 \\ 
\bottomrule
 \end{tabularx}

    \label{tab:Gaussian_parameters}
\end{table*}

For the two-stage decoder, we use a CNN with  $F_1=5$, $S_1=3$, $F_2=8$, and $S_2=2$, determined by cross-validation. 
The CNN is trained with a dataset of size $80000$,
batch set size $300$, and the quantile loss function with $\xi = 0.75$.
The number of iterations of the WMS decoder for each code is the same as above.

Figure~\ref{fig:BER_performance} illustrates the BER vs. SNR for different codes, and different decoders for the same code.
In each of Figure~\ref{fig:BER_performance}a--d, one can compare different decoders at about the same complexity (except for the parallel 
decoder with the largest $\nu$ that shows the smallest achievable BER).
For instance, it can be seen in Figure~\ref{fig:BER_performance}a that
the two-stage decoder achieves half the BER of the WMS with SS-PAN decoder at SNR $6.4$ dB for the
short length code $\mathcal C_1$,
or approximately $0.32$ dB gain in SNR at a  BER of $10^{-4}$.
For this code, the parallel WMS decoders with 3 iterations and $\nu = 9$ outperforms the tanh-based BP with
nearly the same complexity.
Figure~\ref{fig:BER_performance}b,c show that the
two-stage decoder  offers about an order-of-magnitude improvement in the BER compared to the Type $T_a$ WMS decoder  at $4.2$ dB for moderate-length
codes  $\mathcal C_2$ and $\mathcal C_3$, or over $0.1$ dB gain at a BER of $10^{-6}$.
The performance gains vary with the code, parameters, and  SNR.

Figure~\ref{fig:BER_performance}d--f compare decoders with different complexities at about the same performance.
The proposed adaptive model-based decoders achieve the same performance of the model-agnostic static decoders,
with far fewer parameters.

The computational complexity of the decoders are presented in Table~\ref{tab:ComputComplex}.
For the CNN,  from \eqref{eq:complex_CNN}, $\text{RM} = n(8T + 95) - 24T - 310$, which is further reduced by a factor of  4
upon pruning at the threshold $\tau_{\text{prun}} = 0.001$, with minimal impact on BER.
Thus, the two-stage decoder requires less than half of the RM of the WMS decoder with no or Type $T_a$ parameter sharing.
Moreover, the two-stage decoder requires approximately one-fifth of the RM of the parallel decoders with $\nu = 16$.
Compared to the WMS decoder with SS-PAN \cite{M.Lian2019},
the two-stage decoder has nearly double the complexity, albeit with much lower BER, as seen in Figure~\ref{fig:BER_performance}a.

\begin{table*}[t]

    \caption{Computational complexity of decoders, for the AWGN channel.}
\centering
    \begin{tabularx}{0.7\textwidth}{cCCccc}
        \toprule
        \multirow{2.5}{*}{} & \multirow{2.5}{*}{\boldmath{$\gamma^{*(t)}$}} & \multirow{2.5}{*}{\boldmath{$\alpha^{*(t)}$}} & \multicolumn{3}{@{}c@{}}{\textbf{Average RM per Iteration}} \\
        \cmidrule{4-6}
         &  &  & \makecell{\boldmath{$\mathcal{C}_1$}} & \makecell{\boldmath{$\mathcal{C}_2$}} & \makecell{\boldmath{$\mathcal{C}_3$}} \\
        \midrule
         \multicolumn{6}{@{}c@{}}{\text{No weight sharing}} \\
\midrule
          WMS \cite{E.Nachmani18} & $\gamma_{v,c}^{(t)}$ & $1$ & $768$ & 25,792
& 40,160\\
\midrule
        \multicolumn{6}{@{}c@{}}{\text{Weight sharing}} \\
\midrule
         WMS, Type $T_a$ & $\gamma_{v,c}$ & $1$ & $768$ & 25,792 & 40,160\\
\midrule
        \makecell{WMS, Type $T_a V C$} & $\gamma$ & $1$ & $63$ & $3226$ & $4016$\\
\midrule
        \makecell{Parallel WMS\\ Type $T_b V C$, $\nu=16$}& $\gamma^{(t)}$ & $1$ & $1440$ & 77,376 & 84,336 \\
\midrule
        \makecell{Parallel WMS  \\ Type $T_b V C$, $\nu=64$}& $\gamma^{(t)}$ & $1$ & -- & $\simeq 3.09\times {10}^5$ & $\simeq$$3.37\times {10}^5$\\
\midrule
        \makecell{Parallel WMS \\ Type $T_b V C$, $\nu=1024$}& $\gamma^{(t)}$ & $1$ & 92,340 & -- & -- \\
\midrule
        \makecell{Two-stage decoder \\Type $T_b V C$, $\tau_{\text{prun}}=0.001$} & $\gamma^{(t)}$ & $1$ & $\simeq$300 & $\simeq$14,093 & $\simeq$17,558 \\
\midrule
        \makecell{WMS SS$-$PAN,\\Type $T_a V C$ \cite{M.Lian2019}} & $\gamma^{(t)}$ & $1$ & $153$ & $8060$ & $9287$ \\
        \bottomrule
    \end{tabularx}
    \label{tab:ComputComplex}
\end{table*}

\emph{High-rate regime:}
To further investigate the error correction performance of the decoders at high rates,
we consider three single-edge QC-LDPC codes, $\mathcal{C}_{6}(1050, 875)$, $\mathcal{C}_7(1050, 850)$,
and $\mathcal{C}_8(4260, 3834)$,  associated, respectively, with
the PCMs $\matdet{H}_6(\lambda=7, \omega=42, M=25, \matdet{P}_6)$, $\matdet{H}_7(8, 42, 25, \matdet{P}_7)$,
and $\matdet{H}_8(6, 60, 71, \matdet P_8)$. These codes have rates $r=0.84$, $0.81$, and $0.9$, respectively,
and are constructed using the PEG algorithm.
The PEG algorithm requires the degree distributions of the Tanner graph, which are optimized using the stochastic
 extrinsic information transfer (EXIT) chart described in Appendix~\ref{qc-ldpc-prmts}.
Additionally, we include the polar code $\mathcal{C}_{9}(1024, 854)$ with $r=0.84$ from the 5G-NR standard as a state-of-the-art benchmark.
The code parameters, including degree distribution polynomials and the exponent matrices,
are given in Appendix \ref{qc-ldpc-prmts}.

The acquisition function with   active learning in the parallel decoders is based on
Figure~\ref{fig:threshold_C3}. The figure  shows the scatter plot of the
pre-FEC error $\mathfrak e_1=\norm{\vc c-\bar{\vc y}}$ versus post-FEC error $\mathfrak e_2=\norm{\vc c -\hat{\vc c}}_H$, for $340$ examples $(\vc c, \vc y)$
for $\mathcal{C}_6$ at $E_b/N_0=4.25$ dB.
Here, $\bar{\vc y}$ is the hard decision of the LLRs at the channel output defined in \eqref{SyndromCheckHard_1},
and $\hat{\vc c}$ is decoded with the best decoder at epoch $e$, \ie, the WMS with weights from epoch $e-1$.
The acquisition function retains $(\vc c, \vc y)$ if $\mathfrak e_1=0$ (no error) or if $(\mathfrak e_1, \mathfrak e_2)$
falls in the rectangle in Figure~\ref{fig:threshold_C3} (with error).
The rectangle is defined such that $\prob( (\mathfrak e_1,\mathfrak e_2)\in \mathcal S)\geq 0.95$.
It is ensured that 70\% of examples satisfy $\mathfrak e_1=0$ and 30\% with $(\mathfrak e_1,\mathfrak e_2)$ in the rectangle $\mathcal S$.
In this example, $\mathfrak e_1\in\{ 80, 81, \cdots, 100\}$, and $\mathfrak e_2 \in [\mu - 2\sigma,\mu + 2\sigma]$, $\mu=149.97$, $\sigma = 12.2$.
We use $\mathfrak{b}_1=2000$, $\mathfrak{b}_2=20000$, $\mathfrak{b}_s=500$ and the learning rate $0.001$.

\begin{figure}[t]

\includegraphics[width=0.475\textwidth]{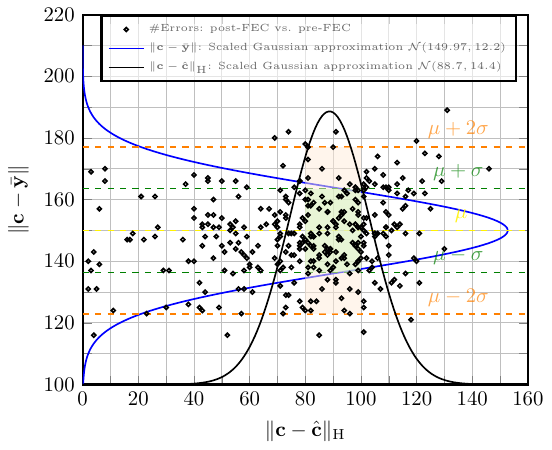}
\caption{The scatter plot of $(\mathfrak e_1, \mathfrak e_2)$ for $\mathcal{C}_{9}$ at  $E_b/N_0=4.25$ dB, for the AWGN channel.
The scaled  Gaussian approximation curve is fitted per axis. }
\label{fig:threshold_C3}
\end{figure}

For the adaptive decoder, we consider five parallel decoders with $T_1=4$ iterations.
The decoder for the binary codes $\mathcal{C}_{6}$, $\mathcal{C}_7$ and $\mathcal{C}_8$ is WMS with Type $T_a V C$ sharing.
The output of the decoder with the smallest syndrome is continued with an MS decoder with $T_2=4$ iterations.
The polar code, however, is decoded with either a  cyclic redundancy check (CRC)
and successive cancellation list (SCL) with list size $L$ \cite{Tal2015}  
or the optimized successive cancellation (OSC) \cite{sural2020tb}.

Figure~\ref{fig:perform_polar_vs_qcl} shows the performance of the adaptive and static MS decoders for $\mathcal{C}_{6}$,
$\mathcal{C}_7$ and $\mathcal{C}_{9}$.
The polar code $\mathcal{C}_{9}$ with 24 CRC bits is simulated using AFF3CT software 
 toolbox \cite[v3.0.2]{AFF3CT}.
It can be seen that at high SNRs, $E_b/N_0 \geq 4.6$, $\mathcal{C}_6$ and $\mathcal{C}_7$ decoded with adaptive parallel decoders outperform $\mathcal{C}_{9}$.
Given this, and the higher complexity of decoding the polar code with either SCL or OSC \cite{Tal2015}, 
the choice of QC-LDPC codes for the inner code for the optical fiber channel in Section \ref{sec:fiber-results} is justified.

\begin{figure}[t]

   \includegraphics[width=0.475\textwidth]{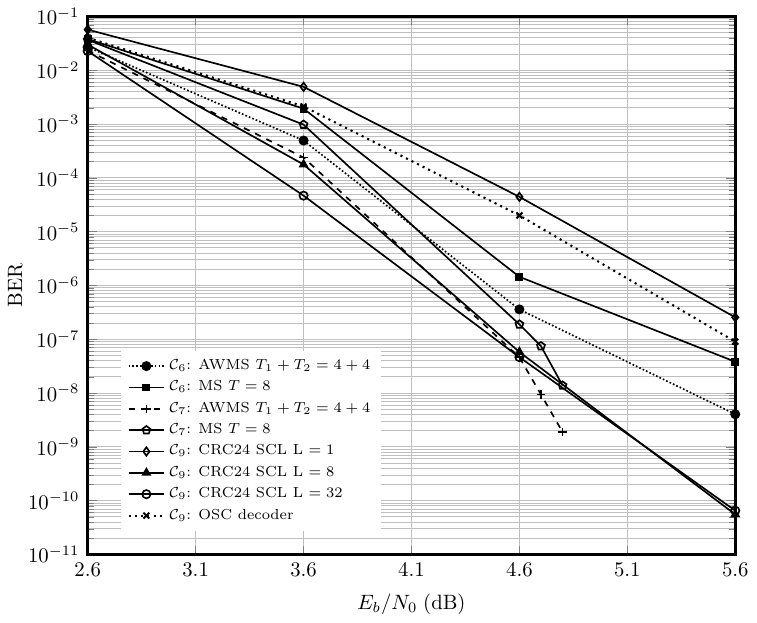}
    \caption{Performance of the polar code $\mathcal{C}_{9}\left(1024,854\right)$ versus
    QC-LDPC codes $\mathcal{C}_{6}\left(1050,875\right)$ and $\mathcal{C}_7\left(1050,850\right)$, for the AWGN channel in the high-rate regime.
    The curve for OSC decoder is from  \cite{sural2020tb}. 
}
    \label{fig:perform_polar_vs_qcl}
\end{figure}

Figure~\ref{fig:perform_qc3} shows the performance of $\mathcal{C}_8(4260,3834)$ with rate $0.9$.
The adaptive WMS decoder with $T_1+T_2=8$ iterations outperforms the static MS decoder with $T=8$ iterations at $E_b/N_0=5$
by an order of magnitude in BER.

\begin{figure}[t]

   \includegraphics[width=0.475\textwidth]{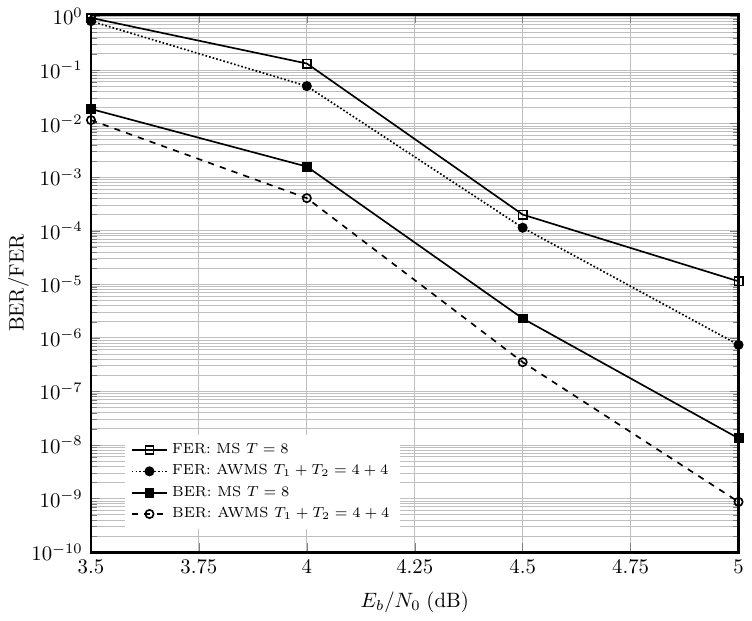}
    \caption{Performance of the static and adaptive  MS decoder for  $\mathcal{C}_8\left(4260,3834\right)$ at $E_b/N_0=4$ dB, for the AWGN channel in the high-rate regime.}
    \label{fig:perform_qc3}
\end{figure}

The gains of WBP depend on parameters such as the block length or SNR \cite{tang2021} (Sec.~IV.~d~\cite{S.Adiga2024}). 
In general, the gain is decreased when the block length is increased, with other parameters remaining fixed.

\subsection{Optical Fiber Channel}
\label{sec:fiber-results}

We simulate a 16-QAM WDM optical communication system described in \mbox{Section \ref{sec:fiber-channel-model}}, with parameters described in Table \ref{tab:SystemPara}.
The continuous-time model \eqref{eq:nls} is simulated with the split-step Fourier method with a spatial step size of  100~m 
and a simulation bandwidth of 200 GHz.
DBP with two samples/symbol is applied to compensate for the physical effects  and to  obtain
the per-symbol channel law $\Pr(y | s)$, $s\in\mathcal A$, $y\in\mathbb{C}$.
For the inner code in the concatenated code, we consider two QC-LDPC codes of rate $r_i=0.92$: binary single-edge code $\mathcal{C}_{10}(4000,3680)$ and non-binary multi-edge code $\mathcal C_{11}(800, 32)$
over $\mathbb{F}_{32}$, respectively, with PCMs $\matdet{H}_{10}(4, 50, 80, \matdet{P}_{10})$ and $\matdet{H}_{11}(2, 25,$ $32,$ $\matdet{P}_{11})$,
given in Appendix~\ref{qc-ldpc-prmts}.
For the component code used in the outer spatially coupled code,
we consider multi-edge QC-LDPC code $\mathcal{C}_{12}(3680,3520)$ with the PCM $\matdet{H}_{12}(1, 23, 160, \matdet{P}_{12})$.
For $m_s = 2$ and $L = 100$, the resulting SC-QC-LDPC code has the PCM $\matdet{H}_{\text{SC}}(1, 23, 160,\matdet P_{12}, 2,100,\Bar{\matdet{B}}_{12})$,
where $\bar{\matdet B}_{12}$ is the spreading matrix.
The outer SC-QC-LDPC code is encoded with the sequential encoder~\cite{Tazoe2012}.
This requires that the top-left
$\lambda M\times \omega M$ block $\matdet{H}_0(0)$ of $\matdet{H}_{\text{SC}}$
in \mbox{Equation (\ref{eq:Htl})}
is of full rank. Thus, $\bar{\matdet B}_{12}$ is designed to fulfill this condition.
In Equation (\ref{eq:Htl}), we have $\matdet{H}_t(\lambda=1,\omega=23, M=160)\eqdef \matdet H\in\FF^{160\times 3680}$, $t=0,1,2$, $\matdet{H}_0$ is of full rank,
and $\matdet{H}_{\text{SC}}\in\FF^{16320\times 368000}$.
The rate of the component code is  $r_{\text{QC}}=1-1/23\approx 0.956$, and the rate of outer SC code
is $r_{\text{o}}$ $= r_{\text{QC}} - \frac{m_{s}\lambda}{L\omega}$ $\approx 0.955$, so $r_{\text{total}}=r_{\text{i}}\cdot r_{\text{o}}\approx 0.88$.
The $\matdet P_{12}$ and $\Bar{\matdet B}_{12}$
matrices are constructed heuristically and are given in  Appendix~\ref{qc-ldpc-prmts}.

The inner code is decoded with the parallel decoder, with five  decoders  with four ~iterations each.
The decoder for the binary code $\mathcal{C}_{9}$ is WMS with Type $T_a V C$ parameter sharing, while for the non-binary code, $\mathcal{C}_{11}$
is WEMS with Type $T_a V C Q$ sharing and $\beta_{c,v, q'}^{(t)}=1$.
The static EMS algorithm \cite{Takasu2024} is parameterized as in Section~\ref{sec:wbp}, initialized with the LLRs  computed from Eq. 1~\cite{Boutillon2013}.
The outer code is decoded with the SWD, with the static MS decoding of a   maximum of $26$ iterations per window.

Tables \ref{tab:fiber-binary-code} and \ref{tab:fiber-nonbinary} contain a summary of the numerical results. $\text{BER}_i$ is pre-FEC BER,
and the reference BER for the coding gain is $\ber_o = {10}^{-12}$.
At $\mathcal{P}=-10$ dBm, the total gap to $\text{NCG}_f$ for the adaptive weighted min-sum AWSM (resp., WEMS) decoder is $2.51$ (resp., $1.75$),
while this value is  $3.31$ (resp., $2.29$) and $3.44$ (resp., $2.69$), respectively, for the NNMS (resp., NNEMS) and MS (resp., EMS) decoders.
Thus, the adaptive WBP provides a coding gain of 0.8 dB compared to the static NNMS decoder with about the same computational complexity and decoding latency.

\begin{table*}[t]
    \caption{Concatenated inner binary QC-LDPC code $\mathcal{C}_{10}$ and outer SC-QC-LDPC code $\mathcal{C}_{12}$ 
    with an $r_{\text{total}}=0.88$ for the optical fiber channel. The sections for $\text{BER}_{\text{i}}$ 0.012 and 0.025 correspond 
    to  average powers -10 and -11 dBm, respectively. NCGs are in dB.}
\centering
    %\begin{tabularx}{\textwidth}{ccCcccccccc}
        \begin{tabularx}{0.75\textwidth}{
  c@{\hskip7pt}
  c@{\hskip7pt}
  C@{\hskip7pt}
  c@{\hskip7pt}
  c@{\hskip7pt}
  c@{\hskip7pt}
  c@{\hskip7pt}
  c@{\hskip7pt}
  c@{\hskip7pt}
  c@{\hskip7pt}
  c
}
\toprule
         & $\textbf{BER}_{\textbf{i}}$ &
    \makecell{\textbf{Inner-SD} \\ \textbf{Decoder}} &
    \makecell{$\textbf{BER}_{\textbf{o}}$ \\ \textbf{Inner}} &
    \makecell{$\textbf{BER}_{\textbf{o}}$ \\ \textbf{Total}} &
    \makecell{\textbf{NCG} \\ \textbf{Inner}} &
    \makecell{\textbf{NCG} \\ \textbf{Total}} &
    \makecell{$\textbf{NCG}_f $ \\ \textbf{Inner}} &
    \makecell{$\textbf{NCG}_f $ \\ \textbf{Total}} &
    \makecell{\textbf{Gap to} $\textbf{NCG}_f$ \\ \textbf{Inner}} &
    \makecell{\textbf{Gap to } $\textbf{NCG}_f$ \\ \textbf{Total}} \\
        \midrule
           & \multirow{5}{*}{$0.012$} &  AWMS & $3.29\times 10^{-6}$
 & $4.52\times 10^{-8}$ & $5.64$ & $6.93$  & $9.38$ &  $9.44$ & $3.74$ &$2.51$\\
           \cmidrule{3-11}
           & & \makecell{\text{NNMS}\\ $\theta = 0.75$}& $4.02\times 10^{-6}$ & $5.43\times 10^{-7}$~ & $5.56$  & $6.13$ & $9.38$ & $9.44$ & $3.82$ &$3.31$\\
            \cmidrule{3-11}
           & &  MS & $4.77 \times 10^{-6}$ & $7.75\times 10^{-7}$ & $5.49$ & $6.00$ &  $9.38$ &  $9.44$ & $3.89$ & $3.44$\\
        \midrule
           & \multirow{5}{*}{$0.025$} &  AWMS &  $0.019$ & $0.017$  &  $0.13$  & $0.12$  & $10.20$ &  $10.28$ & $10.07$  &$10.16$\\
           \cmidrule{3-11}
           & & \makecell{\text{NNMS}\\ $\theta = 0.72$}& $0.02$ & $0.018$ & $0.04$ & $0.03$ & $10.20$ & $10.28$ & $10.16$  &$10.25$\\
            \cmidrule{3-11}
           & & MS & $0.023$ & $0.02$ & $-0.20$ & $-0.15$ & $10.20$ & $10.28$ & $10.40$  &$10.43$\\
        \bottomrule
    \end{tabularx}
    \label{tab:fiber-binary-code}
\end{table*}

\begin{table*}[t]

    \caption{Concatenated inner non-binary QC-LDPC code $\mathcal{C}_{11}$ and outer SC-QC-LDPC code $\mathcal{C}_{12}$ with $r_{\text{total}}=0.88$ for the optical fiber channel.
The sections for $\text{BER}_{\text{i}}$ 0.012 and 0.025 correspond to  average powers -10 and -11 dBm, respectively. NCGs are in dB.}
%    \begin{tabularx}{\textwidth}{ccCcccccccc}
\centering
\begin{tabularx}{0.75\textwidth}{
  c@{\hskip7pt}
  c@{\hskip7pt}
  C@{\hskip7pt}
  c@{\hskip7pt}
  c@{\hskip7pt}
  c@{\hskip7pt}
  c@{\hskip7pt}
  c@{\hskip7pt}
  c@{\hskip7pt}
  c@{\hskip7pt}
  c
}

\toprule
         & $\textbf{BER}_{\textbf{i}}$ &
    \makecell{\textbf{Inner-SD} \\ \textbf{Decoder}} &
    \makecell{$\textbf{BER}_{\textbf{o}}$ \\ \textbf{Inner}} &
    \makecell{$\textbf{BER}_{\textbf{o}}$ \\ \textbf{Total}} &
    \makecell{\textbf{NCG} \\ \textbf{Inner}} &
    \makecell{\textbf{NCG} \\ \textbf{Total}} &
    \makecell{$\textbf{NCG}_f $ \\ \textbf{Inner}} &
    \makecell{$\textbf{NCG}_f $ \\ \textbf{Total}} &
    \makecell{$\textbf{Gap to } \textbf{NCG}_f$ \\ \textbf{Inner}} &
    \makecell{$\textbf{Gap to } \textbf{NCG}_f$ \\ \textbf{Total}} \\
\midrule
         & \multirow{5}{*}{$0.012$} &  AWEMS & ~$3.21\times 10^{-8}$~ &~$2.74\times 10^{-9}$~ & $7.23$ & $7.69$ & $9.38$ &  $9.44$ & $2.15$ & $1.75$ \\
           \cmidrule{3-11}
           & & \makecell{ \text{NNEMS}\\ $\theta = 0.2$}& ~$4.61\times 10^{-8}$~ & ~$2.11\times 10^{-8}$~ & $7.12$ & $7.15$ & $9.38$ & $9.44$ & $2.26$ & $2.29$ \\
            \cmidrule{3-11}
           & &  EMS & $2.44\times 10^{-7}$ & $8.20\times 10^{-8}$ & $6.60$ & $6.75$ &  $9.38$ &  ²$9.44$ & $2.78$ & $2.69$\\
\midrule
           & \multirow{5}{*}{$0.025$} &  AWEMS & $0.0063$ & $0.0051$ & $1.73$ & $1.80$ & $10.20$ &  $10.28$ & $8.47$ & $8.48$\\
           \cmidrule{3-11}
           & & \makecell{\text{NNEMS}\\ $\theta = 0.25$}& $0.0087$ &  $0.0075$ & $1.32$ & $1.32$ & $10.20$ & $10.28$ & $8.88$ & $8.96$ \\
            \cmidrule{3-11}
           & & EMS & $0.025$ & $0.022$ & $-0.36$ & $-0.32$ & $10.20$ & $10.28$ & $10.56$ & $10.60$\\
\bottomrule
    \end{tabularx}
    \label{tab:fiber-nonbinary}
\end{table*}

\section{Conclusions}
\label{sec:conclusions}

Adaptive decoders are proposed for codes on graphs that can be decoded with message-passing algorithms. 
Two variants, the parallel WBP and the two-stage decoder, are studied.
The parallel decoders search for the best sequence of weights in real time using multiple instances of the WBP decoder running concurrently, while the
two-stage neural decoder employs an NN to dynamically determine the weights of WBP for each received word.
The performance and complexity of the adaptive and several static decoders  are compared for a number of codes over an AWGN and optical fiber channel.
The simulations show that significant improvements in BER can be obtained using adaptive decoders, depending on the channel, SNR, the code and  its parameters.
Future work could explore further reducing the computational complexity of the online learning, 
and applying adaptive decoders to other types of codes and wireless channels.

\appendices 

\section{Protograph-Based QC-LDPC Codes}
\label{subsec:qc-ldpc}

In this appendix and the next, we provide the supplementary information necessary to reproduce the results presented in this paper.
The presentation in Appendix \ref{sec:sc-code}  may be of independent interest, as it provides an accessible exposition of the construction and
decoding of the SC codes.

\subsection{Construction for the Single-Edge Case}
\label{subsubsec:qc-ldpc}

A single-edge PB QC-LDPC code $\mathcal{C}$  is constructed in two steps.
First, a base matrix  $\matdet{B} \in\mathbb{F}_2^{\lambda\times\omega}$ is constructed, where $\lambda,  \omega\in\mathbb{N}$, $\lambda \leq \omega$.
Then, $\matdet{B}$ is expanded to the PCM of $\mathcal C$ by replacing each zero in $\matdet{B}$ with the all-zero matrix $\textbf{0}\in\mathbb{F}_2^{M\times M}$, where $M\in\mathbb N$ is the lifting factor, and  a one in row $i$ and column $j$
with a sparse circulant matrix $\matdet{H}_{i,j} \in\mathbb{F}_2^{M\times M}$.

Let $\matdet{P}\in [-1{:}M-1]^{\lambda\times\omega}$ be the exponent matrix of the code,
with the entries denoted by $p_{i,j}$.
The matrices $\matdet{H}_{i,j}$ (and $\matdet{B}$) can be obtained from $\matdet P$ as follows.
Denote by $\matdet{I}^{n} \in\mathbb{F}_2^{M\times M}$, $n\in [-1{:}M-1]$ the circulant permutation matrix obtained by cyclic-shifting of each row of the $M\times M$ identity matrix $n$ positions to the right, with the convention that $\matdet{I}^{-1}$ is the all-zero matrix.
Then, $\matdet{H}_{i,j}=\matdet{I}^{p_{i,j}}$.
The PCM
$\matdet{H}:= \matdet{H}(\lambda, \omega, M, \matdet{P})
\in\mathbb{F}_2^{\omega M\times\lambda M }$ of this QC-LDPC code is
\begin{IEEEeqnarray}{rCl}
\label{eq:H-block}
\matdet{H}=
\begin{pmatrix}
 \matdet{I}^{p_{1,1}}  &  \matdet{I}^{p_{1,2}}  & \cdots   & \matdet{I}^{p_{1,\omega}} \\
  \matdet{I}^{p_{2,1}} &   \matdet{I}^{p_{2,2}} & \cdots   &  \matdet{I}^{p_{2,\omega}} \\
 \vdots   & \vdots    & \ddots   & \vdots \\
  \matdet{I}^{p_{ \lambda,1}}  &   \matdet{I}^{p_{\lambda,2}} & \cdots   &  \matdet{I}^{p_{\lambda,\omega}}\\
        \end{pmatrix}.
\end{IEEEeqnarray}
This code has a length $n = \omega M$, $k=\omega M-r$, and a rate $r  = 1-r/(\omega M)$, where  $r\leq \lambda M$ is the rank of $\matdet H$. If $\matdet H$ is of a full rank, $r=1-\lambda/\omega$.
The base matrix is also obtained from $\matdet P$, as $\matdet B_{i,j}=0$ if $p_{i,j}=-1$ and $\matdet B_{i,j}=1$ if $p_{i,j} \neq -1$.

Denote the Tanner graph of $\mathcal{C}$ by  $T_{\mathcal{C}}$, and let $d_{c}$ and $d_{v}$ be, respectively, the degree of the check node $c$ and the variable node $v$ in $T_{\mathcal{C}}$.
If $\matdet B$  is regular (i.e., its rows have the same Hamming weight), then $\matdet H$ is regular,
and the check and variable nodes of $\mathcal C$ have the same degrees $d_c$ and $d_v$, respectively.
In this case, $\mathcal{C}$ is said to be $\left(d_c,d_v\right)$-regular.
More generally,
a variable node's degree distribution polynomial can be defined as $\Upsilon(x)=\sum_{d} \Upsilon_d x^{d-1}$,
where $\Upsilon_d$ is the fraction of variable nodes of degree $d$,
and a check node degree distribution $\Lambda(x)=\sum_{d} \Lambda_d x^{d-1}$, where $\Lambda_d$ is the
fraction of check nodes of degree  $d$.

The parameter matrices $\matdet B$ and $\matdet P$ can be obtained so as to maximize the girth of $T_{\mathcal{C}}$
using search-based methods such as the PEG algorithm \cite{Hu2001,LiZo2004},
algebraic methods \cite[Sec. 10]{ryan2009}, or a combination of them \cite{tasdighi2021integer}.

\begin{example}\label{EX_1}
Consider $\lambda=3$, $\omega=5$, and the base matrix $\matdet B\in\mathbb{F}_2^{3\times 5}$
   \begin{IEEEeqnarray*}{rCl}
   \matdet B &=& \begin{pmatrix}
         0& 1& 1& 1& 0  \\
         1& 1& 0& 1& 1 \\
         1& 0& 1& 1& 1
         \end{pmatrix}.
\end{IEEEeqnarray*}
For any exponent matrix $\matdet P$,  $\matdet H\in\mathbb{F}_2^{3M\times 5M}$ is
   \begin{IEEEeqnarray*}{rCl}
      \matdet H &=& \begin{pmatrix}
         \matdet{I}^{-1}& \matdet{I}^{p_{1,2}}& \matdet{I}^{p_{1,3}}& \matdet{I}^{p_{1,4}}& \matdet{I}^{-1}  \\
         \matdet{I}^{p_{2,1}}& \matdet{I}^{p_{2,2}}& \matdet{I}^{-1}& \matdet{I}^{p_{2,4}}& \matdet{I}^{p_{2,5}} \\
         \matdet{I}^{p_{3,1}}& \matdet{I}^{-1}& \matdet{I}^{p_{3,3}}& \matdet{I}^{p_{3,4}}& \matdet{I}^{p_{3,5}}
    \end{pmatrix}.
    \qed
\end{IEEEeqnarray*}
\end{example}

\subsection{Construction for the Multi-Edge Case}
The above construction can be extended to multi-edge PB QC-LDPC codes.
Here, $\matdet B_{i,j}$, instead of a binary number, is a sequence of length $D_{i,j}$ with entries in $\FF$.
Likewise, $p_{i,j}$ is a sequence of length $D_{i,j}$ with entries $p_{i,j}^{d}\in [-1{:}M-1]$, $d\in[D_{i,j}]$.
Then, $\matdet{H}_{i,j}=\sum_{d=1}^{D_{i,j}} \matdet{I}^{p_{i,j}^{d}}$. This code is represented by a Tanner graph
where there are multiple edges of different types between the variable  and the check nodes.
We say this code has type $D= \max_{i, j} D_{i,j} \geq 1$. A single-edge PB QC-LDPC code is type 1.

\begin{example}\label{EX_2}
Consider $\lambda=2$, $\omega=3$,
   \begin{IEEEeqnarray*}{rClrCl}
    \matdet B &=& \begin{pmatrix}
         0& (1,1)& 1  \\
         1& 0& (1,1,1)
         \end{pmatrix}, \text{ and}
         \\
    \matdet P &=& \begin{pmatrix}
         -1& \Bigl(p_{1,2}^{1},p_{1,2}^{2} \Bigr)& p_{1,3}  \\
         p_{2,1}& -1& \Bigl(p_{2,3}^{1},p_{2,3}^{2},p_{2,3}^{3} \Bigr)
         \end{pmatrix}.
\end{IEEEeqnarray*}
Then
   \begin{IEEEeqnarray*}{rCl}
        \matdet H &=& \begin{pmatrix}
         \matdet{I}^{-1}& \matdet{I}^{p_{1,2}^{1}}+\matdet{I}^{p_{1,2}^{2}}& \matdet{I}^{p_{1,3}}  \\
         \matdet{I}^{p_{2,1}}& \matdet{I}^{-1}&  \matdet{I}^{p_{2,3}^{1}}+\matdet{I}^{p_{2,3}^{2}}+\matdet{I}^{p_{2,3}^{3}}
    \end{pmatrix}.
    \qed
\end{IEEEeqnarray*}
\end{example}

\subsection{Construction for Non-binary Codes}
In the non-binary codes, the entries of codewords, parity-check, and generator matrices are in a finite
field $\mathbb{F}_q$, where $q$ is a power of a prime.
There are several ways to construct non-binary PB QC-LDPC codes.
The base matrix $\matdet B \in \mathbb{F}_2^{\lambda \times \omega}$ typically remains binary, as defined as in Appendix~\ref{subsubsec:qc-ldpc}.
We use the unconstrained and random assignment strategy in~\cite{Dolecek2014} to extend $\matdet B$ to a PCM.

There is significant flexibility in selecting the edge weights for constructing non-binary QC-LDPC codes, which can be classified 
as constrained or unconstrained \cite[Sec. II]{Dolecek2014}. In this work, we focus on an unconstrained and random assignment strategy, where each edge in $T_{\mathcal{C}}$ corresponding to a $1$ in the binary matrix $B$ is replaced with a coefficient $h_{i,j} \in \mathbb{F}_q$. Alternatively, these coefficients could be selected based on predefined rules to ensure appropriate edge-weight diversity, which can lead to enhanced performance. This methodology allows non-binary codes to retain the structural benefits of binary QC-LDPC codes while extending their functionality to finite fields, offering improved error-correcting capabilities for larger $q$ values.

\subsection{Encoder and Decoder}
The generator matrix of the code is obtained by applying Gaussian elimination in the binary field to \eqref{eq:H-block}. The encoder is then implemented efficiently using shift registers \cite{Zongwang2006}.

The QC-LDPC codes are typically decoded using belief propagation (BP), as described in Section \ref{sec:wbp}.

\section{Spatially Coupled LDPC Codes}
\label{sec:sc-code}

\subsection{Construction}
\label{Outer_code}

The SC-QC-LDPC codes in this paper are constructed based on the edge spreading process \cite{mitchell17}.
Denote the PCM of the constituent PB QC-LDPC code by $\matdet{H}(\lambda, \omega, M, \matdet P)$.
Denote the PCM of the corresponding SC-QC-LDPC code by $ \matdet{H}_{\text{SC}} \allowbreak:=\allowbreak \matdet{H}_{\text{SC}}\allowbreak \left(\lambda,\allowbreak\omega, M, \matdet P,  m_s, L, \allowbreak\bar{\matdet B}\right)$, with the additional parameters of the syndrome memory $m_s\in\mathbb N$, coupling length $L\allowbreak\in \allowbreak\mathbb N$, and the spreading matrix $\bar{\matdet B}\in[-1{:}m_s]^{\lambda\times\omega}$.
Then, $\matdet{H}_{\text{SC}}\allowbreak\in \allowbreak\mathbb{F}_2^{\lambda M(m_s+L+1)\times \omega M L}$ is given by
\begin{center}
\includegraphics[width=0.5\textwidth]{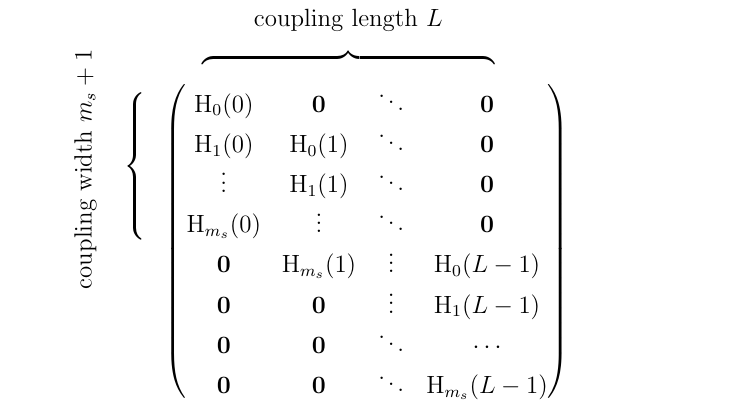}
\end{center}
in which $\matdet{H}_{t}(\ell),\mathbf{0} \in \mathbb{F}_2^{\lambda M\times \omega M}$ are
$\lambda\times\omega$ block matrices, $t\in[0:m_s]$, $\ell\in[0,L-1]$ and $\matdet{H}_{t}(\ell)$ is obtained from $\bar{\matdet B}$ as
\begin{IEEEeqnarray}{rCl}
    \matdet{H}_{t}(\ell) &&
\text{~at row-block $i$ and column-block $j$}
\nonumber\\
&=&
\begin{cases}
    \matdet{I}^{p_{i, j}}, & \text{if~}\bar{\matdet B}_{i, j} =t,\\
    \mathbf{0}, & \text{otherwise},
\end{cases}
\label{eq:Htl}
\end{IEEEeqnarray}
where $\matdet{I}^{p_{i, j}},\mathbf{0}\in \mathbb{F}_2^{M\times M}$.

If $\matdet{H}_{t}\left({\ell}_1\right)=\matdet{H}_{t}\left({\ell}_2\right)$, $\forall t\in [0,m_s]$ and $\forall \ell_1 , \ell_2\in [0,L-1]$,
the code is time-invariant. If $\matdet H$ and  $\matdet{H}_{\text{SC}}$ are full-rank,
then the rate of SC-QC-LDPC code is $r= r_{\text{QC}} - \frac{m_{s}\lambda}{L\omega}$,
where $r_{\text{QC}} = 1 - \frac{\lambda}{\omega}$
is the rate of the component QC-LDPC code.
For a fixed $m_s$, as  $L\rightarrow \infty$, then $r\rightarrow r_{\text{QC}}$.
Thus, the rate loss in SC-LDPC codes can be reduced by increasing the coupling length.

\begin{example}
\label{EX_3}
Consider any QC-LDPC code, $m_s = 2$, $L=4$ and the spreading matrix
\begin{IEEEeqnarray*}{rCl}
\bar{\matdet B}=
\begin{pmatrix}
    -1& 1& 0& 2& -1  \\
1& -1& 1& 0& -1 \\
1& 2& -1& 0& 2
\end{pmatrix}
\end{IEEEeqnarray*}
Then, $\matdet{H}_2(\ell)$ is obtained by replacing each entry of  $\bar{\matdet B}$ that is $2$ at row $i$ and column $j$ with $\matdet{I}^{p_{i,j}}$,
and other entries with $\mathbf 0$. Thus, for all $\ell\in[0:3]$
\begin{IEEEeqnarray*}{rCl}
\matdet{H}_{2}(\ell)&=&
    \begin{pmatrix}
    \mathbf{0} & \mathbf{0}& \mathbf{0}& \matdet{I}^{p_{1,4}}& \mathbf{0}  \\
    \mathbf{0}& \mathbf{0}& \mathbf{0}& \mathbf{0}& \mathbf{0} \\
    \mathbf{0} & \matdet{I}^{p_{3,2}}& \mathbf{0}& \mathbf{0}& \matdet{I}^{p_{3,5}} 
    \end{pmatrix}.
\end{IEEEeqnarray*}
In a similar manner,
\begin{IEEEeqnarray*}{rClrCl}
  \matdet{H}_{0}(\ell) &=&
 \begin{pmatrix}
    \mathbf{0} & \mathbf{0}& \matdet{I}^{p_{1,3}}& \mathbf{0}& \mathbf{0}  \\
    \mathbf{0} & \mathbf{0}& \mathbf{0}& \matdet{I}^{p_{2,4}}& \mathbf{0} \\
    \mathbf{0} & \mathbf{0}& \mathbf{0}& \matdet{I}^{p_{3,4}}& \mathbf{0} 
    \end{pmatrix},
        \\
   \matdet{H}_{1}(\ell)&=&
   \begin{pmatrix}
    \mathbf{0} & \matdet{I}^{p_{1,2}}& \mathbf{0}& \mathbf{0}& \mathbf{0}  \\
    \matdet{I}^{p_{2,1}}& \mathbf{0}& \matdet{I}^{p_{2,3}}& \mathbf{0}& \mathbf{0} \\
    \matdet{I}^{p_{3,1}}& \mathbf{0}& \mathbf{0}& \mathbf{0}& \mathbf{0} \\
    \end{pmatrix}.
\end{IEEEeqnarray*}
If $\matdet{H}_{\text{SC}}$ is full-rank, then $r = r_{\text{QC}} - \frac{2\cdot 3}{4 \cdot 5}$.
If $\matdet{H}_t(\ell)$ is also full-rank, then $r_{\text{QC}}=0.4$ and $r = 0.1$. \qed
\end{example}

The SC-LDPC code is  efficiently  encoded sequentially \cite{Tazoe2012} so that at each
spatial position $\ell$, $(\omega - \lambda)M$ information bits are encoded out of $(\omega - \lambda) ML$.

If an entry of $\matdet B$ corresponding to an $\matdet{H}_i(\ell)$ is a sequence, the corresponding entry in $\Bar{\matdet B}$ is also a sequence. Thus, if $\matdet{H}_i(\ell)$ represents a multi-edge code for some $i$ or $\ell$, so does $\matdet{H}_{\text{SC}}$.

\begin{figure*}[t]
  \centering
    \begin{tabular}{c@{~~~~~~~~~~~}c}
  \includegraphics[width=0.28\textwidth]{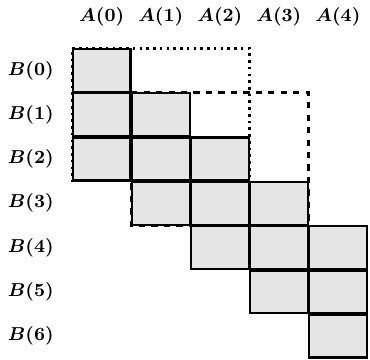} &
    \includegraphics[width=0.5\textwidth]{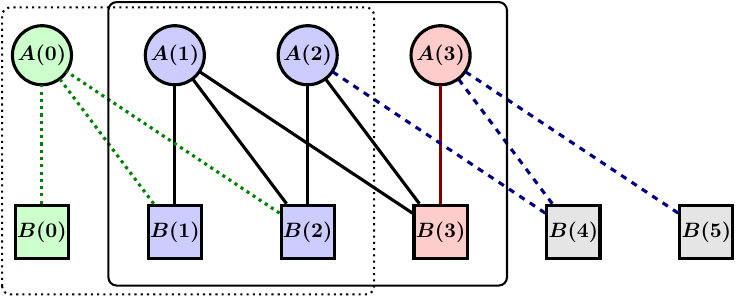}\\[1em]
  (a)  & (b)
  \end{tabular}
  \caption{%
(a) 
A window in SWD is a $m_w\times m_w$ block, each block with size $\lambda M\times \omega M$. The window slides diagonally by one block.
The variable nodes in the block at the block-position $i\times j$ are denoted by $\matdet A(i)$, and the check nodes by $\matdet B(j)$.
Here, $m_s=2$ and $m_w=3$.
(b) The message update for the current window $\ell$ is denoted by the solid rectangle.
The variable nodes in the previous windows $\ell'<\ell$ and not in the current window $\ell$, denoted in green,
send fixed messages $\matdet L^{(T_{\ell-1}, \ell-1)}_{v2c}$ to the check nodes in the current window in any iterations $t$.
The blue variable nodes in the current and a previous window send $\matdet L^{(T_{\ell-1}, \ell-1)}_{v2c}$
to the check nodes in the current window at iteration $t= 1$,
which would be updated according to \eqref{eq:check-node-update}--\eqref{eq:var-node-update} for $t \geq 2$.
The red variable nodes in the current window, but not in any previous window, send $\matdet L(\vc y)$ to the check nodes
in the current window at $t= 1$, which would be updated at $t \geq 2$.
The edges from the gray check nodes in windows $\ell'>\ell$  are discarded for the decoding at position $\ell$.}
  \label{fig:WinDec}
\end{figure*}

\subsection{Encoder}
The SC-LDPC codes are encoded with the sequential encoder \cite{Tazoe2012}.

\subsection{Sliding Window Decoder}
Consider the SC LDPC code $\matdet{H}_{\text{SC}}\allowbreak(\lambda,\allowbreak\omega, M,\allowbreak \matdet{P}, m_s, L, \allowbreak\bar{\matdet B})$ in Appendix
 \ref{Outer_code}.
Note that any two variable-nodes in the Tanner graph of the code whose corresponding columns in $\matdet{H}_{\text{SC}}$ are at least
$\left(m_s + 1\right)M\omega$ columns apart do not share any common check-nodes. Thus, they are not involved in the same parity-check equation.
The SWD uses this property and runs a local BP decoder on windows of $\matdet{H}_{\text{SC}}$ shown in  Figure~\ref{fig:WinDec}.

SWD works through a sequence of spatial iterations $\ell$, where a rectangular window slides from the top-left to the bottom-right side of  $\matdet{H}_{\text{SC}}$.
In general, a window matrix $\matdet{H}_{w}$ of size $m_w$ consists of $m_w \lambda M$ consecutive rows
and  $m_w \omega M$ consecutive columns in $\matdet{H}_{\text{SC}}$. At each iteration $\ell$, it moves $\lambda M$ rows down and $\omega M$
columns to the right in $\matdet{H}_{\text{SC}}$. Thus, a window is an $m_w\times m_w$ block, starting from the top left and moving diagonally one block down per iteration.
There is a special case where the window reaches the boundary.
The way the windows near the boundary are terminated impacts performance \cite{Lentmaier2011,ali2017improving}.
Our setup for window termination at boundary is \textit{early termination}, which is  discussed in section III-B1~\cite{ali2017improving}.

Denote the variable nodes in window $\ell$ by
\[
V(\ell)= \Bigl\{ v_i\in V: i= (\ell-1)\omega M, \cdots, (m_{w}+\ell-1)\omega M-1\Bigr\}.
\]
The check-nodes directly connected to $V(\ell)$ are  $C(\ell)$.
Define $\tilde{V}(\ell) = \cup_{\ell'<\ell}\bigl\{v\in V(\ell')\setminus V(\ell), \: v\text{~connected to~}C(\ell) \bigr\}$
(the variable nodes in the previous windows not in the current window, shown in green in Figure~\ref{fig:WinDec})
$\bar{V}(\ell) = \cup_{\ell'<\ell}\bigl\{v\in V(\ell')\cap V(\ell), \: v\text{~connected to~}C(\ell) \bigr\}$
(the variable nodes in the current window and any previous window, shown in blue in Figure~\ref{fig:WinDec}),
and $\hat{V}(\ell) = \cup_{\ell'<\ell}\bigl\{v\in V(\ell)\setminus V(\ell'),$ $ v\text{~connected to~}C(\ell) \bigr\}$
(the variables nodes in the current window not in any previous window, shown in red in Figure~\ref{fig:WinDec}).
Let $\matdet L^{(t,\ell)}_{v2c}$ and $\matdet L^{(t,\ell)}_{c2v}$ be LLRs in window $\ell$ and  iteration $t \in[T_{\ell}]$ in BP.
At $t=1$, the BP is initialized as
\begin{IEEEeqnarray}{rCl}
\matdet L^{(1, \ell)}(v)=
\begin{cases}
\matdet L^{(T_{\ell-1}, \ell-1)}(v) , & v\in \tilde{V}(\ell)\cup \bar{V}(\ell),\\
\matdet L\left(\vc y\right), &  v\in \hat{V}(\ell)
\end{cases}
\IEEEeqnarraynumspace
\label{eq:SC-update}
\end{IEEEeqnarray}
The update equation for  the variable node is
\begin{IEEEeqnarray*}{rCl}
\matdet L^{(t, \ell)}_{v2c}=
\begin{cases}
\matdet L^{(T_{\ell-1}, \ell-1)}_{v2c}, & v\in \tilde{V}(\ell),\\
\matdet L\left(\vc y\right){+}\hspace*{-5mm} \sum\limits_{c' \in \matdet{C}_{v} \setminus \{c\}}
\hspace*{-4mm}\matdet{L}^{(t-1),\ell}_{c'2v} , & v\in \bar{V}(\ell) \\
\matdet L\left(\vc y\right), &  v\in \hat{V}(\ell).
\end{cases}
\IEEEeqnarraynumspace
\end{IEEEeqnarray*}
The update relation for  $\matdet L^{(t,\ell)}_{c2v}$ is given by \eqref{eq:check-node-update}, with no weights, applied for $c\in C(\ell)$ and  $v\in V(\ell)$.
After $T_{\ell}$ iterations, the variables in the window $\ell$, called target symbols, are decoded.
The SWD is illustrated in Figure~\ref{fig:WinDec}.

\section{Parameters of Codes}
\label{qc-ldpc-prmts}

For  low-rate codes, first,
the degree distributions of the Tanner graph are determined using the extrinsic information transfer (EXIT) chart \cite{Land2014}.
The EXIT chart produces accurate results if  $n\ra\infty$ \cite{Land2014}.
For  high-rate codes, we apply the stochastic EXIT chart, which in the short block length regime
yields better coding gains compared to the deterministic variant \cite{Koike20}. For instance, while the EXIT chart suggests that
the degree distributions of $\mathcal{C}_6$ are optimal
 near $E_b/N_0 = 3$ dB,
the stochastic EXIT chart in Figure~\ref{fig:StEXIT Chart_C12} suggests $E_b/N_0 = 3.5$ dB.
Indeed, at $E_b/N_0 = 3$ dB, the check-node extrinsic information $I_C$ intersects with variable-node extrinsic information $I_V$ for the deterministic EXIT chart.
The exponent matrices are obtained using the PEG algorithm, which takes the optimized degree distribution polynomials.

\begin{figure}[t]

\includegraphics[width=0.457\textwidth]{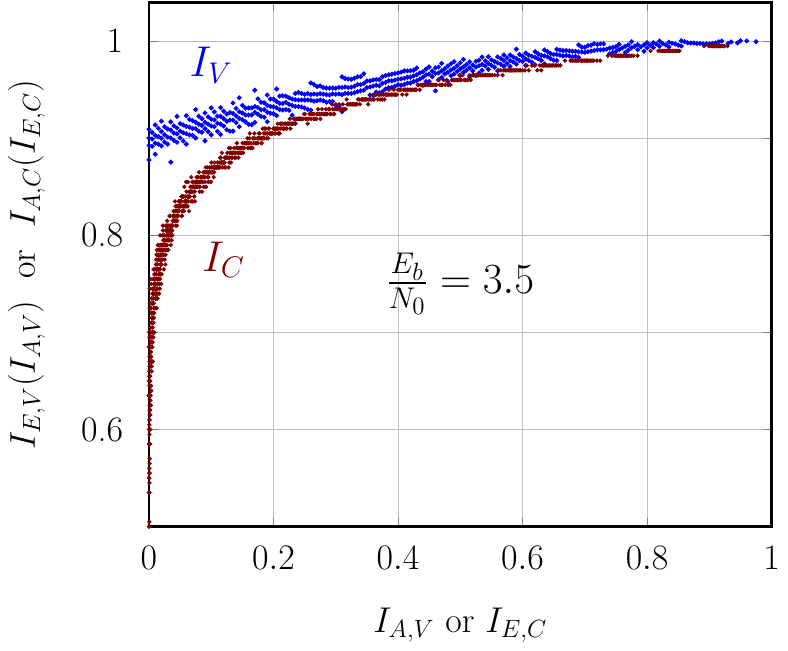}
    \caption{Stochastic EXIT chart for the high-rate code $\mathcal{C}_{6}$ at ${E_b}/{N_0}=3.5$ dB, for the AWGN channel.}
        \label{fig:StEXIT Chart_C12}
\end{figure}

The matrices below are vectorized row-wise. They can be unvectorized considering their dimensions.

\def\yshift{2mm}
\def\yyyshift{3pt}
\newcommand{\bme}{\\&&}
\newcommand{\mywidth}{0.5\textwidth}

\noindent $\mathcal{C}_2$:    $\lambda = 4$, $\omega = 8$, $M = 403$, $\Lambda(x)=x^3$, $\Upsilon(x)=x^{7}$,  $\matdet P_2$ below  

\noindent $[
345\;152\;72\;376\;377\;197\;4\;144\;187\;398\;320\;225\;330\;198\;79 \\ 289 \; 271\;165\;259\;105\;288\;254\;51  \;236\;111\;233\;380\;332\;47\;76\\ 222\;247]$

\vspace{0.5\yshift}
\noindent\rule{40mm}{0.85pt}
\vspace{0.5\yshift}

\noindent $\mathcal{C}_3$: $\lambda = 5$, $\omega = 16$, $M = 251$, $\Lambda(x)=x^4$, $\Upsilon(x)=x^{15}$,   $\matdet P_3$ below 

\vspace{\yyyshift}

\noindent $[
6\;98\;208\;177\;76\;76\;76\;48\;111\;76\;76\;34\;76\;76\;64\;85\;198\;42\\
155\;127\;29\;32\;35\;10\;76\;44\;47\;8\;53  \;56\;47\;71\;31\;211\;158\;0\\
238\;111\;199\;8\;195\;248\;121\;167\;46\;170\;246\;140\;117\;51\;3\;65\\
57\;\:150\;\:243 \; \;57\;213\;20\;113\;164\;48\;141\;222\;85\;181\;142\;121\\
210\;229\;98\;218\;59\;242\;76\;196\;23\;185\;54\;162\;52 ]$

\vspace{0.5\yshift}
\noindent\rule{40mm}{0.85pt}
\vspace{0.5\yshift}

\noindent $\mathcal{C}_6$: $\lambda = 7$, $\omega = 42$, $M = 25$, $\Lambda(x)=0.714x^2+0.286x^3$, $\Upsilon(x)=0.857x^{18}+0.143x^{23}$,  $\matdet P_6$ below

\vspace{\yyyshift}

\noindent $\bm{[}0\;0\;0\;0\;0\;0\;0\;0\;0\;0\;0\;0\;0\;0\;0\;0\;0\;0\;0\;0\;0\;0\;0\;0\;-1\;-1\;-1\;-1\;-1\;-1\;-1\;-1\;-1\;-1\;-1\;-1\;-1\;-1\;-1\;-1\;-1\;-1\;1\;3\;8\;18\;7\;17\;22\;24\;-1\;-1\;-1\;-1\;-1\;16\;-1\;-1\;-1\;-1\;-1\;-1\;14\;-1\;-1\;-1\;-1\;19\;4\;-1\;-1\;22\;-1\;8\;-1\;-1\;5\;-1\;-1\;9\;0\;14\;-1\;15\;3\;12\;7\;11\;14\;18\;13\;22\;-1\;-1\;-1\;-1\;-1\;-1\;-1\;17\;-1\;-1\;-1\;-1\;-1\;-1\;15\;-1\;-1\;3\;-1\;18\;-1\;-1\;-1\;9\;8\;-1\;-1\;19\;9\;20\;-1\;21\;5\;-1\;13\;-1\;-1\;-1\;-1\;-1\;-1\;-1\;2\;6\;16\;11\;14\;9\;19\;23\;12\;-1\;-1\;-1\;-1\;-1\;-1\;-1\;20\;-1\;19\;-1\;-1\;-1\;17\;-1\;-1\;-1\;8\;16\;18\;4\;3\;-1\;17\;-1\;-1\;-1\;15\;-1\;-1\;-1\;-1\;-1\;6\;24\;14\;22\;3\;11\;1\;19\;-1\;-1\;13\;-1\;-1\;-1\;-1\;-1\;18\;-1\;18\;-1\;15\;-1\;-1\;-1\;-1\;13\;11\;19\;-1\;-1\;20\;17\;-1\;21\;-1\;-1\;-1\;-1\;17\;-1\;-1\;-1\;-1\;13\;-1\;-1\;-1\;-1\;-1\;-1\;3\;9\;24\;4\;21\;1\;16\;22\;14\;-1\;-1\;14\;23\;24\;11\;-1\;16\;23\;-1\;-1\;-1\;-1\;-1\;-1\;7\;23\;-1\;-1\;-1\;-1\;-1\;-1\;18\;-1\;-1\;-1\;-1\;15\;-1\;-1\;-1\;-1\;9\;11\;21\;8\;17\;4\;14\;16\;-1\;11\;-1\;6\;11\;13\;13\;20\;6\;13\;-1\;-1\;16\;-1\;-1\;-1\;-1\;-1\bm{]}$

%%%%%%%%%%%%%%%%
\vspace{0.5\yshift}
\noindent\rule{40mm}{0.85pt}
\vspace{0.5\yshift}

\noindent $\mathcal{C}_7$:  $\lambda = 8$, $\omega = 42$, $M = 25$, $\Lambda(x)=0.596x^2+0.404x^3$, $\Upsilon(x)=0.125x^{9}+0.125x^{16}+0.5$ \\ 
$x^{17}+0.125x^{19}+0.125x^{23}$,  $\matdet P_7$ below   

\vspace{\yyyshift}

\noindent $\bm{[}0\;0\;0\;0\;0\;0\;0\;0\;0\;0\;0\;0\;0\;0\;0\;0\;0\;0\;0\;0\;0\;0\;0\;0\;-1\;-1\;-1\;-1\;-1\;-1\;-1\;-1\;-1\;-1\;-1\;-1\;-1\;-1\;-1\;-1\;-1\;-1\;1\;3\;8\;18\;7\;17\;22\;24\;-1\;-1\;-1\;-1\;-1\;16\;-1\;-1\;-1\;-1\;-1\;-1\;14\;-1\;-1\;-1\;-1\;20\;-1\;16\;-1\;-1\;-1\;-1\;4\;-1\;10\;-1\;-1\;20\;19\;-1\;2\;9\;3\;12\;7\;11\;14\;18\;13\;22\;-1\;-1\;-1\;-1\;-1\;-1\;-1\;17\;-1\;-1\;-1\;-1\;-1\;-1\;15\;-1\;-1\;-1\;-1\;22\;6\;-1\;-1\;13\;-1\;21\;4\;-1\;-1\;-1\;-1\;14\;1\;17\;13\;-1\;-1\;-1\;-1\;-1\;-1\;-1\;2\;6\;16\;11\;14\;9\;19\;23\;12\;-1\;-1\;-1\;-1\;-1\;-1\;-1\;-1\;-1\;-1\;5\;-1\;-1\;3\;-1\;10\;-1\;23\;-1\;-1\;8\;4\;-1\;2\;-1\;-1\;-1\;15\;-1\;-1\;-1\;-1\;-1\;6\;24\;14\;22\;3\;11\;1\;19\;-1\;-1\;13\;-1\;-1\;-1\;-1\;-1\;20\;-1\;-1\;-1\;-1\;13\;-1\;-1\;-1\;-1\;0\;3\;3\;-1\;12\;3\;7\;-1\;-1\;-1\;-1\;-1\;17\;-1\;-1\;-1\;-1\;13\;-1\;-1\;-1\;-1\;-1\;-1\;3\;9\;24\;4\;21\;1\;16\;22\;-1\;-1\;12\;-1\;6\;-1\;-1\;6\;-1\;-1\;-1\;-1\;0\;6\;18\;19\;-1\;6\;-1\;-1\;-1\;-1\;-1\;-1\;18\;-1\;-1\;-1\;-1\;15\;-1\;-1\;-1\;-1\;9\;11\;21\;8\;17\;4\;14\;16\;23\;1\;3\;-1\;24\;10\;0\;-1\;-1\;23\;-1\;0\;-1\;-1\;-1\;23\;-1\;19\;-1\;-1\;-1\;-1\;-1\;-1\;-1\;-1\;-1\;-1\;-1\;-1\;-1\;-1\;-1\;-1\;-1\;-1\;-1\;-1\;-1\;-1\;-1\;-1\;8\;5\;20\;-1\;-1\;7\;10\;24\;21\;10\;-1\;11\;22\;-1\;-1\;-1\;-1\;-1\bm{]}$

%%%%%%%%%%%%%%%%
\vspace{0.5\yshift}
\noindent\rule{40mm}{0.85pt}
\vspace{0.5\yshift}

\noindent   $\mathcal{C}_8$: $\lambda = 6$, $\omega = 60$, $M = 71$, $\Lambda(x)=0.1x+0.634x^2+0.266x^3$, $\Upsilon(x)=0.166x^{27}+0.668$ \\
$x^{30}+0.166x^{37}$,  $\matdet P_8$ below 

\vspace{\yyyshift}

\noindent $\bm{[}0\;0\;0\;0\;0\;0\;0\;0\;0\;0\;0\;0\;0\;0\;0\;0\;0\;0\;0\;0\;0\;0\;0\;0\;0\;0\;0\;0\;0\;0\;0\;0\;$ 
$-1\;-1\;-1\;-1\;-1\;61\;-1\;-1\;-1\;-1\;-1\;-1\;-1\;-1\;-1\;-1\;-1\;-1\;-1\;-1\;62\;-1\;-1\;20\;60\;24\;37\;-1\;1\;17\;45\;52\;57\;62\;63\;68\;70\;54\;26\;19\;14\;9\;8\;3\;0\;-1\;-1\;-1\;-1\;-1\;-1\;28\;-1\;-1\;-1\;-1\;-1\;-1\;-1\;-1\;-1\;46\;-1\;-1\;-1\;-1\;34\;51\;-1\;-1\;59\;-1\;18\;-1\;-1\;-1\;61\;-1\;52\;46\;21\;48\;-1\;-1\;44\;34\;-1\;2\;48\;6\;1\;17\;40\;29\;11\;67\;23\;65\;70\;54\;31\;42\;60\;4\;-1\;0\;-1\;-1\;-1\;-1\;2\;-1\;-1\;-1\;-1\;-1\;-1\;-1\;-1\;-1\;46\;-1\;-1\;-1\;-1\;-1\;21\;-1\;23\;-1\;-1\;7\;-1\;-1\;53\;-1\;-1\;62\;-1\;32\;27\;27\;39\;1\;-1\;15\;-1\;16\;0\;-1\;-1\;-1\;-1\;-1\;-1\;16\;-1\;-1\;-1\;-1\;-1\;-1\;-1\;-1\;3\;51\;64\;14\;29\;44\;47\;62\;68\;20\;7\;57\;42\;27\;24\;9\;-1\;-1\;-1\;32\;29\;-1\;-1\;-1\;-1\;38\;-1\;34\;-1\;45\;-1\;67\;-1\;-1\;60\;40\;-1\;37\;-1\;9\;20\;-1\;69\;20\;-1\;0\;-1\;-1\;-1\;-1\;21\;-1\;-1\;-1\;-1\;-1\;-1\;-1\;-1\;-1\;2\;18\;3\;51\;49\;16\;33\;59\;69\;53\;68\;20\;22\;55\;38\;12\;-1\;-1\;62\;-1\;-1\;-1\;-1\;48\;0\;37\;48\;-1\;26\;19\;59\;60\;49\;38\;-1\;-1\;-1\;-1\;67\;-1\;-1\;-1\;24\;-1\;-1\;-1\;0\;1\;40\;42\;-1\;-1\;-1\;-1\;-1\;-1\;-1\;-1\;-1\;-1\;-1\;-1\;67\;14\;48\;55\;-1\;-1\;-1\;-1\;-1\;-1\;-1\;-1\;-1\;-1\;25\;67\;16\;41\;64\;0\;68\;66\;50\;20\;63\;67\;34\;45\;0\;7\;29\;11\;11\;-1\;-1\;-1\;61\;-1\;-1\;-1\;-1\;-1\bm{]}$

%%%%%%%%%%%%%%%%
\vspace{0.5\yshift}
\noindent\rule{40mm}{0.85pt}
\vspace{0.5\yshift}

\noindent  $\mathcal{C}_{10}$:   $\lambda = 4$, $\omega = 50$, $M = 80$, $\Lambda(x)=0.02+0.18x\:+\:0.64x^2\:+\:0.16x^3$, $\Upsilon(x)=0.25x^{29}$ \\
$0.25x^{33}+0.25x^{34}+0.25x^{47}$,   $\matdet P_{10}$ below

\vspace{\yyyshift}

\noindent $\bm{[}0\;0\;0\;0\;0\;0\;0\;0\;0\;0\;0\;0\;0\;0\;0\;0\;0\;0\;0\;0\;0\;0\;0\;0\;$
$0\;0\;0\;0\;0\;0\;0\;0\;0\;0\;0\;0\;0\;0\;0\;0\;0\;0\;0\;0\;0\;0\;0\;0\;-1\;-1\;1\;2\;8\;-1\;53\;70\;-1\;78\;72\;-1\;27\;-1\;3\;6\;-1\;-1\;79\;50\;77\;-1\;-1\;-1\;-1\;30\;7\;14\;-1\;-1\;51\;-1\;73\;66\;-1\;-1\;29\;-1\;18\;36\;64\;-1\;74\;-1\;62\;44\;16\;60\;-1\;20\;37\;-1\;3\;67\;78\;41\;31\;26\;77\;13\;-1\;39\;49\;54\;9\;-1\;-1\;43\;-1\;-1\;71\;-1\;6\;37\;-1\;-1\;21\;69\;66\;47\;57\;22\;59\;11\;14\;33\;23\;58\;-1\;-1\;44\;18\;-1\;68\;-1\;-1\;36\;62\;-1\;12\;-1\;79\;42\;-1\;-1\;23\;-1\;-1\;38\;79\;-1\;57\;-1\;-1\;-1\;3\;70\;69\;54\;-1\;-1\;77\;10\;11\;26\;68\;-1\;7\;30\;1\;-1\;28\;-1\;73\;50\;-1\;-1\;52\;36\;18\;20\;14\;4\;72\;44\;62\;60\;66\;76\;8\;34\;-1\bm{]}$

%%%%%%%%%%%%%%%%
\vspace{0.5\yshift}
\noindent\rule{40mm}{0.85pt}
\vspace{0.5\yshift}

\noindent $\mathcal{C}_{11}$:  $\lambda = 2$, $\omega = 25$, $M = 32$, $\Lambda(x)=x$, $\Upsilon(x)=x^{24}$,  $\matdet P_{11}$ below

\vspace{\yyyshift}

\noindent $[(0,17)\;(-1)\;(0,20)\;(-1)\;(0,21)\;(-1)\;(0)\;(0)\;(0)\;(0)\;(0)\;(0) \\
(0)\;(0)\;(0)\;(0)\;(0)\;(0)\;(0)\;(0) \;(0)\;(0)\;(0)\;(0)\;(0)\:\;(-1)\;(0,17) \\
(-1)\;\:(0,20)\;(-1)\;\:(0,21)\;(0)\;(1)\;(2)\;(3)\;(4)\;(5)\;(6)\;(7)\;(8) \\ 
\;(9)\;(10)\;(16)\;(25)\;(26)\;(27)\;(28)\;(29)\;(30)\;(31) ]$

%%%%%%%%%%%%%%%%
\vspace{0.5\yshift}
\noindent\rule{40mm}{0.85pt}
\vspace{0.5\yshift}

\noindent   $\mathcal{C}_{12}$:  $\lambda = 1$, $\omega = 23$, $M = 160$, $\Lambda(x)\:=\:0.480x\: +\: 0.130x^2\: +\: 0.217x^3, 0.130x^4\: +\: 0.043x^5$\\
 $\Upsilon(x)=x^{71}$, $\matdet P_{12}$ below  

\vspace{\yyyshift}

\noindent $[(1,151,151,55,127,151)\;(138,27,139,144)\;\:(88,57,-1) \\
(111,-1,47)\;\:(130,15,33)\;(11,47, 118,15,108)\;\:(109,5,-1) \\
(143,-1,140)\;\:(100,14,14,141,-1)\;\:(12,-1,20)\;\:(91,42,96) \\
(74,-1,72) \;(54,29,155,157,159)\;(83,82,-1)\;(0,77,141,78,13)\\
(112,-1,59)\;(119,74,56)\;(48 ,6,55,157,-1)\;(85,-1,9) \\ 
(41,80,121,2,-1)\;(103,-1,45)\;(60,117,52,87)\;(99,148,-1)]$

\vspace{\yshift}
\noindent  $m_s=2$, $\Bar{\matdet B}_{12}$ below

\vspace{\yyyshift}

\noindent $[(0,0,1,2,2,2)\;(0,1,1,2)\;(0,1,-1)\;(0,-1,2)\;(0,1,2) \\ 
(0,0,0,1,2) (0,1,-1)\;(0,-1,2)\;(0,1 ,1,1,-1)\\
(0,-1,2) (0,1,2) (0,-1,2) (0,1,2,2,2)\;(0,1,-1) \\
(0,0,0,1,2) (0,-1,2) (0,1,2) (0,1,1,1,-1), (0,-1,2) \\
(0,0,0,1,-1) (0,-1,2) (0,1,2,2) (0,1,-1) ]$

%\bibliographystyle{IEEEtran}

%\bibliography{Refs.bib} 

%%%%%% bbl file %%%%%%%%%%%%%%%%%%%%%%%%%%%%%%%%%%%%%


%%%%%%%%%%%%%%%%%%%%%%%%%%%%%%%%%%%%%%%%%%%%%%%%%%%%%%

\end{document}